\def\year{2019}\relax
\documentclass[letterpaper]{article} %DO NOT CHANGE THIS
\usepackage{aaai19}  %Required
\usepackage{times}  %Required
\usepackage{helvet}  %Required
\usepackage{courier}  %Required
\usepackage{url}  %Required
\usepackage{graphicx}  %Required
\nocopyright
\usepackage{xcolor} %for comments only
\usepackage{gensymb} %for degree - in the acknowledge men
\usepackage{amsmath}
\usepackage{dsfont}
\usepackage{amsthm}
\usepackage{multirow} %merge rows in the table
\usepackage{mathtools}

\newcount\Includeappendix %include appendices or not

\newtheorem{claim}{Claim}
\newtheorem{definition}{Definition}
\newtheorem{theorem}{Theorem}
\newtheorem{proposition}{Proposition}
\newtheorem{lemma}{Lemma}

\theoremstyle{definition}

\DeclareSymbolFont{symbolsC}{U}{txsyc}{m}{n}
\DeclareMathSymbol{\notniFromTxfonts}{\mathrel}{symbolsC}{61}
\DeclarePairedDelimiter\abs{\lvert}{\rvert}%
\newcommand{\defeq}{\stackrel{\textnormal{def}}{=}}
\newcommand{\tupbracket}[1]{\left\langle {#1} \right\rangle}
\newcommand{\ind}{\mathds{1}}
\newcommand\bl[1]{\boldsymbol{ #1 } }

\usepackage[linesnumbered,ruled,noend]{algorithm2e}
\usepackage{booktabs}       % professional-quality tables

\usepackage{amssymb}
\usepackage{subfig}
\usepackage{kbordermatrix}
\usepackage{booktabs}    % professional-quality tables

\title{Convergence of Learning Dynamics in Information Retrieval Games}

\author{
  Omer Ben{-}Porat, Itay Rosenberg \and Moshe Tennenholtz \\
  Technion - Israel Institute of Technology\\
  Haifa 32000 Israel\\
  \{omerbp@campus, itayrose@campus, moshet@ie\}.technion.ac.il
}

\begin{document}
%\onecolumn

\maketitle

\begin{abstract}
We consider a game-theoretic model of  information retrieval with strategic authors. We examine two different utility schemes: authors who aim at maximizing exposure and authors who want to maximize active selection of their content (i.e., the number of clicks). We introduce the study of author learning dynamics in such contexts. We prove that under the probability ranking principle (PRP), which forms the basis of the current state-of-the-art ranking methods, any better-response learning dynamics converges to a pure Nash equilibrium.  We also show that other ranking methods induce a strategic environment under which such a convergence may not occur.
\end{abstract}

\section{Introduction}
\label{sec:intro}
Information retrieval is probably the most central task carried out by consumers and users of on-line media. The basic information retrieval task involves ranking documents in a corpus by their relevance to the information needs expressed in a query. In adversarial retrieval settings such as the Web, information resources (contents) are owned by strategic bodies - website owners (henceforth \textit{authors}). Authors can strategically change their content in order to improve their rankings in response to a query in a practice referred to as search engine optimization (SEO) \cite{GyongyiG05}. 
Therefore, the authors are players in a game, altering their content to increase their \textit{utility}: increase exposure of their content (in a plain content setting) or to increase selection of their content (``clicks'' in a sponsored content setting). In this strategic game, the search engine serves as a \textit{mediator} between users and authors, and attempts to match queries and websites.

Despite the tremendous amount of work on information retrieval and SEO published during past decades, mathematical modeling of the aforementioned strategic behavior has only been formally suggested and studied recently \cite{Basat+al:15a,basat2017game,raifer2017information}. One central question in this regard is whether learning dynamics, whereby at every step one author alters her content to increase her utility, is likely to converge. Convergence would suggest that authors should only invest a considerably limited amount of time altering their websites until their utility cannot be further improved. An accompanying question is whether such convergence occurs when state-of-the-art approaches to information retrieval, aiming at ranking documents in the corpus according to estimated relevance probabilities with respect to a given query, are used.
The basis for all such retrieval methods is the probability ranking principle (PRP) \cite{Robertson:77a}.
%One example of such approach is the seminal work of Robertson on the probability ranking principle (PRP) \cite{Robertson:77a}.

In this paper we introduce what is, to the best of our knowledge, the first attempt to explore the learning dynamics of strategic behavior in information retrieval systems such as the Web, through a formal theoretical model. Our main result proves that under the PRP, any better-response learning dynamics converges to a pure Nash equilibrium. This result is obtained for the two prevalent utility schemes: authors seeking content exposure (i.e., exposure-targeted), and authors seeking to increase ``clicks'' in content selection (i.e., action-targeted). Interestingly, this learning dynamics convergence property, which rarely exists in games, is obtained even though our class of games are not potential games \cite{MondererShapley96}. We also show that other plausible ranking methods may not induce such convergence, which further highlights the significance of our results.

\subsection{Related Work} 
The concept of mediators in strategic environments is widely known to the game-theory community \cite{ashlagi2009mediators,Aumann74,MTSME}, and the design of a mediator (or in a different terminology, a mechanism) is often called mechanism design \cite{NisanRonenstoc}. In the context of information retrieval, a search engine can be viewed as a mediator between two parties: users and authors.

Considering strategic behavior in an information retrieval context is the aim of \citeauthor{basat2017game} \shortcite{basat2017game}. The work of \citeauthor{basat2017game} presents a game-theoretic approach to information retrieval, and illustrates that the myopic static view falls short in dynamic and adversarial settings. \citeauthor{basat2017game} explicitly assume that users will select the highest ranked result, a somewhat strong assumption but nevertheless justified by a large body of empirical work \cite{Butman+al:13a,JoachimsGPHG05,Liu+Wei:16a,ghose2012mobile}. Note that in this case, PRP coincides with ranking the most relevant document highest. \citeauthor{basat2017game} analyze the user social welfare, defined as the quality of documents available in the presence of strategic behavior of the authors. Interestingly, they demonstrate that introducing randomization into a ranking function can sometimes lead to social welfare that transcends that of applying the PRP. In this paper we also adopt the game-theoretic approach to information retrieval, but explore a different criterion, which is the learning dynamics in games induced by the selection of the PRP as the mediator. Furthermore, beyond the action-targeted utility suggested in \citeauthor{basat2017game}, we also analyze exposure-targeted utility.

\citeauthor{ben2018game} \shortcite{ben2018game} consider mediator design in recommendation systems with strategic content providers. They highlight several fairness-related properties that a mediator should arguably satisfy, along with the requirement of pure Nash equilibrium existence. They claim against PRP, as they show that in their mathematical model the PRP mediator (termed TOP in their work) satisfies the fairness-related properties, but may lead to a game without pure Nash equilibria and hence without better-response convergence. However, their mathematical model differs from the one in this paper, since e.g. they allow the mediator to present an empty list of documents, which is highly unlikely in information retrieval settings. 

Designing a mediator for improved social welfare was recently proposed by \citeauthor{rstfl} \shortcite{rstfl}, who also make the connection between recommendation systems and facility location games \cite{Hotelling}. In their model as well, matching users with their nearest facility may yield a low social welfare in case the content providers are strategic. Their goal is to design a mediator that optimizes welfare in equilibrium and does not intervene too much.

In this work, however, we do not study the social welfare, but rather focus on the \textit{learning dynamics}. Learning dynamics is an important concept in machine learning and game theory \cite{CesaBianchi,claus1998dynamics,FreundSchapire,palaiopanos2017multiplicative,syrgkanis2015fast,meir2010convergence,lev2012convergence}, and work on learning dynamics in games is considered instrumental, e.g., to understanding ad auctions \cite{cary2014convergence}. Better-response learning dynamics are appealing to the (algorithmic) game theory community, as they only assume a minimal form of rationality: under any given profile, a player will act to increase her individual utility.
However, general techniques for showing better-response learning convergence in games are rare, and are based typically on coming up with a potential function \cite{MondererShapley96}, see e.g. \cite{garg2016learning,palaiopanos2017multiplicative,ben2018competing}. However, as exact potential functions imply the games are congestion games \cite{rosenthal73}, it is easy to observe that our games do not fit that category.

Another interesting class of games which are not potential games for which better-response dynamics always converge is \cite{Milchtaich96}. However, that setting is quite remote from ours, as in \citeauthor{Milchtaich96}'s work the players share a common set of strategies.
\subsection{Our Contribution}
Our main conceptual contribution is the explicit analysis of learning dynamics in information retrieval systems that is motivated by strategic behavior. Our demonstration of convergence serves as an important justification for the use of the PRP, and should be taken into account when designing stable and robust information retrieval systems.

The key technical contribution of this paper is the proof that under PRP any better-response dynamics converges to a pure Nash equilibrium. We prove this claim for both exposure-targeted and action-targeted utility schemes. As stated above, the convergence of better-response learning dynamics in our setting is obtained although the class of games we consider do not have an exact potential function. Moreover, we show that other ranking methods induce a strategic environment under which such convergence may not occur.
Together, our results provide strong novel game-theoretic justification to the PRP and illustrate its applicability in an adversarial context such as the Web.

\subsection{Paper Organization}
The rest of the paper is organized as follows. Section \ref{sec:problem} formalizes the model we adopt, as well as an informal introduction to the relevant core game-theoretic concepts and an illustrative example. In Section \ref{sec:PRP} we analyze better-response learning with the PRP mediator for both utility schemes. 
In Section \ref{sec:nonlearnability} we show non-learnability of mediators other than the PRP, and Section \ref{sec:disc} is devoted to  discussion and future work. Due to space limitations, some of the proofs of this paper are deferred to the supplementary material. 

\section{Problem Statement} 
\label{sec:problem}

An authors game is composed of a set of \textit{authors} $N=[n] \defeq \{1,2,\ldots,n\}$, each owning one document/website/blog. $M=[m]$ is the set of \textit{topics}, and we assume both $n$ and $m$ are finite. An author's pure strategy space is the set of all topics, i.e., she can choose to write her document on any topic. We further assume that each document is concerned with a single topic. The set of all pure strategy profiles is denoted by $A=M^n$, and each strategy profile $\bl a=(a_1,\dots a_n)$ corresponds to a set of documents. A query distribution $D$ over $M$ is publicly known, where each query symbolizes the user mass associated with that topic. Given a topic $k$, we denote by $D(k)$ the demand for topic $k$. 
We further assume w.l.o.g. that $D(1)\geq D(2)\geq \ldots \geq D(m)$.
That is, the topics are sorted according to the query distribution mass in a non-increasing order.

The matrix $Q\in [0,1]^{n\times m}$ is the \textit{quality matrix}, where $Q_{j,k}$ represents the quality for author $j$'s document if she decides to write on topic $k$. This modeling allows an author to have remarkable aptitude for one topic and poor aptitude for another. For example, an economic guru is able to write about sports, but his writing quality w.r.t. sports is substantially lower than economics. 

The function $R$ is the \textit{mediator}, which plays the role of a ranking function or a search engine. The mediator ranks the documents selected by the authors w.r.t. a given query (or equivalently, a topic). We assume for simplicity that users always read the document ranked first. This assumption is consistent with many applications, e.g. the use of personal assistants in mobile devices, where only the first ranked item is shown to the user. Thus, we let $R(Q, k, \bl a)$ denote a distribution over the set of documents selected under $\bl a$  w.r.t. a topic $k\in M$, which represents the probability of being displayed in the first position. For ease of notation, we shall also denote $R_j(Q, k, \bl a)$ as the probability that author $j$ is ranked first under the distribution $R(Q, k, \bl a)$.

The last component $u$ is the \textit{utility function}, which maps every strategy profile to a real-valued vector of length $n$. In this paper, we consider two different utility functions which are motivated by current applications.

Under the exposure-targeted utility, denoted by $u^{Ex}$ , an author's utility is the number of impressions her document receives. Formally,
\begin{definition}[Exposure-targeted utility]
The exposure-targeted utility of author $j$ under a strategy profile $\bl a$ is given by
\[
u^{Ex}_j(\bl a)\defeq \sum_{k=1}^m \ind_{a_j=k} \cdot D(k)\cdot R_j(Q,k,\bl a).
\]
\end{definition}
Note that $u^{Ex}$ depends solely on the user mass of the topic she writes on and the probability of the mediator displaying her document. The other utility function is the action-targeted utility, denoted by $u^{Ac}$. 
\begin{definition}[Action-targeted utility]
The action-targeted utility of author $j$ under a strategy profile $\bl a$ is given by
\[
u^{Ac}_j(\bl a)\defeq \sum_{k=1}^m \ind_{a_j=k} \cdot D(k)\cdot R_j(Q,k,\bl a)\cdot Q_{j,k}. 
\]
\end{definition}
Namely, an author's utility is the user mass of her selected topic times the probability she is ranked first times the quality of her document.

Overall, an authors game can be represented as a tuple $\mathcal G = \tupbracket{N,M,D,Q,R,u}$.

It is convenient to quantify the following; given a strategy profile $\bl a$, let 
$B_k(\bl a)$ denote the highest quality of a document on topic $k$ , i.e.,
\[
B_k(\bl a)\defeq \max\limits_{1 \leq j \leq n}\{Q_{j,k}\cdot \ind_{a_j=k} \}.
\]
Moreover, we denote by $H_k(\bl a)$ the number of authors whose documents have the highest quality among those who write on topic $k$ under $\bl a$,
\[
H_k(\bl a)\defeq\abs{ \{j  \mid j\in [n], Q_{j,k}\cdot \ind_{a_j=k}=B_k(\bl a) \}}.
\]

Unless stated otherwise, we analyze games with a particular mediator, which is based on the PRP. Since we restrict the ranking list to include one rank only, the PRP coincides with ranking first the highest quality document on that topic. We denote by $R^{PRP}$ the mediator that displays the document with the highest quality. In case there are several documents with the highest quality, $R^{PRP}$ ranks first each one of them with equal probability. Formally,

\begin{definition}[The PRP Mediator]
Given a quality matrix $Q$, a topic $k$ and a strategy profile $\bl a$, the $R^{PRP}$ ranks first the document of each author $j$ with a probability of
\[
R_j^{PRP}(Q,k,\bl a) \defeq 
\begin{cases}
\frac{1}{H_k(\bl a)} & Q_{j,k}\cdot \ind_{a_j=k}=B_k(\bl a)\\
0 & \text{otherwise}
\end{cases}.
\]
\end{definition}

\subsection{Further Game Theory Notation}
We now informally introduce some basic game theory concepts used throughout this paper.
For an action profile $\bl a = (a_1,\ldots,a_j,,\ldots,a_n) \in A$, we denote by $\bl a_{-j} = (a_1,\ldots,a_{j-1},a_{j+1},\ldots,a_n)\in A_{-j}$ the action profile of all authors except author $j$. 
A strategy $a_j'\in A_j $ is called a \textit{better response} of author $j$ w.r.t. a strategy profile $\bl a$ if $u_j(a_j', \bl a_{-j}) > u_j(\bl a)$. Similarly, $a_j'\in A_j$ is said to be a \textit{best response} if $u_j(a_j', \bl a_{-j}) \geq \max_{a_j\in A_j}u_j(a_j,\bl a_{-j})$ . We say that a strategy profile $\bl a$ is a \textit{pure Nash equilibrium} (herein denoted PNE) if every author plays a best response under $\bl a$. 

Given a strategy profile $\bl a \in A$, an \textit{improvement step} is a profile $(a_j',\bl a_{-j})$ such that $a_j'$ is a better response of author $j$ w.r.t. $\bl a$.
An \textit{improvement path} $\gamma=(\bl a^1, \bl a^2,\dots )$ is a sequence of improvement steps, where the improvements can be performed by different authors. Namely, in any improvement step along the improvement path exactly one author deviates from the strategy she selected in the previous step, but different authors can deviate in different steps. When the path $\gamma$ is clear from the context, we denote by $p_r$ the author that improves in step $r$. 
Since the number of strategy profiles is finite, every infinite improvement path must contain an improvement cycle.
A non-cooperative game $\mathcal G$ has the \textit{finite improvement property} (FIP for brevity) if 
all the improvement paths are finite; in such a game 
every better-response dynamics converges to a PNE \cite{MondererShapley96}.

\subsection{An Illustrative Example}
\label{subsec:example}

\begin{figure*}
     \centering

     \subfloat[][exposure-targeted]{
     $
     \kbordermatrix{
     & \text{topic } 1 & \text{topic } 2 & \text{topic } 3  \\
     \text{topic } 1 & 0,0.5 & 0.5,0.3 & 0.5,0.2  \\
     \text{topic } 2 & 0.3,0.5 & 0.15,0.15 & 0.3,0.2\\
     \text{topic } 3 & 0.2,0.5 & 0.2,0.3 & 0.2,0
     },\quad 
     $
     \label{<figure1>}}
     \subfloat[][action-targeted]{
      $
     \kbordermatrix{
     & \text{topic } 1 & \text{topic } 2 & \text{topic } 3  \\
     \text{topic } 1 & 0,0.45 & 0.05,0.12 & 0.05,0.04   \\
     \text{topic } 2 & 0.12,0.45 & 0.06, 0.06 & 0.12,0.04  \\
     \text{topic } 3 & 0.16, 0.45 & 0.16,0.12 & 0.16,0
     }
     $
     \label{<figure2>}}
          \caption{     \label{example:utilities} The normal form games induced by the example in Subsection \ref{subsec:example}. Subfigure (a) represents the utilities of author 1 (row) and author 2 (column) under the exposure-targeted utility function, while Subfigure (b) represents the utilities under action-targeted utility function.}
\end{figure*}
To further clarify our notation and setting, we provide the following example. Consider a game with $n=2$ authors, $m=3$ topics, a query distribution mass $D$ such that $D(1)=0.5, D(2)=0.3,D(3)=0.2$, a quality matrix
\[
Q=
\begin{pmatrix}
0.1 & 0.4 & 0.8 \\
0.9 & 0.4 & 0.2 \\
\end{pmatrix},
\]
and $R^{PRP}$ as the mediator.
%Recall that the cell $i,k$ corresponds to the quality author $i$ can obtain we she writes on the $k$'th topic. 
Given the utility function, the induced game can be viewed as a normal form bi-matrix game, as presented in Figure 
\ref{example:utilities}.

First, consider the exposure-targeted utility function. Consider the strategy profile $(a_1,a_2)=(2,2)$. Under this strategy profile the two authors write on topic 2, and their quality on that topic is the same, i.e., $Q_{1,2}=Q_{2,2}=0.4$; thus, $R_1(Q,2,(2,2))=R_2(Q,2,(2,2))=0.5$ and  
\[
u^{Ex}_1(2,2)=u^{Ex}_2(2,2)=\frac{D(2)}{2}=0.15.
\]
Notice that author 2 can improve her utility by deviating to topic 1, i.e., to the strategy profile $(2,1)$. Indeed, this is an improvement step w.r.t. $(2,2)$. In this case, her utility is $u^{Ex}_2(2,1)=0.5$. Clearly $(2,1)$ is a PNE of this game.

The action-targeted utility function induces a different bi-matrix game. The reader can verify that under this utility scheme, the unique PNE is $(3,1)$.

\section{Better-Response Learning with the PRP Mediator}
\label{sec:PRP}
In this section we show that under the PRP mediator, every better-response dynamics converges to a PNE, for both utility schemes. To make this claim more concrete, we use the following definition.
\begin{definition}
\label{def:leanable}
We say that a mediator $R$ is $u$-learnable if every game induced by $R$ and the utility function $u$ has the FIP property.
\end{definition}
Clearly, if any game that consists of $(R,u)$ has the FIP property, then the authors can learn a PNE using any better-response dynamics. We use the above definition to crystallize our goals for this section: we wish to show that $R^{PRP}$ is both $u^{Ex}$-learnable and $u^{Ac}$-learnable. Namely, in Subsection \ref{subsec:Exposure-Targeted-Utility} we show that under the PRP mediator and the exposure-targeted utility function, every improvement path is finite. In Subsection \ref{subsec:Impression-Focused} we prove the equivalent statement for the action-targeted utility function.

Before we go on, we claim that the class of games induced by the PRP mediator does not have an exact potential. 
\begin{proposition}
\label{prop:pot}
The class of games induced by $R^{PRP}$ and either one of $u^{Ex}$ or $u^{Ac}$ does not have an exact potential.
\end{proposition}
\begin{proof}[Proof sketch of Proposition \ref{prop:pot}]
We show that the necessary condition for the existence of an exact potential \cite{MondererShapley96} does not hold for a general authors game with $n\geq 3$ authors. This result is obtained for both utility schemes.
\end{proof}

As mentioned in Section \ref{sec:intro} above, showing the convergence of any better-response dynamics in the lack of exact potential is challenging, and is nevertheless our goal for the rest of this section. In light of that, we shall introduce a further notation. 

\begin{definition}
Given a finite improvement path $\gamma=(\bl a^1,\dots \bl a^l)$, we define 
\[
W_k(\gamma)\defeq\min\limits_{1 \leq r \leq l}\{H_k(\bl a^r)\},
\] 
i.e., $W_k(\gamma)$ is the minimal number of authors writing documents with the highest quality on topic $k$. 
\end{definition}
Note that the minimum is taken over all steps in $\gamma$.

\subsection{Exposure-Targeted Utility}
\label{subsec:Exposure-Targeted-Utility}
We now focus on games with $R^{PRP}$  and $u^{Ex}$, namely the PRP mediator and the exposure-targeted utility function. We show that every improvement path is finite, suggesting that any better-response dynamics converges. The proof of this convergence relies on several supporting claims.

The following Proposition \ref{prop:leqgeq} claims that in every improvement step, the improving author writes with a quality of at least the highest quality obtained in the preceding improvement step, on that particular topic.
\begin{proposition}
\label{prop:leqgeq}
Let $\gamma$ be a finite improvement path, and let $a^{r+1}_{p_r}=k$ for an arbitrary improvement step $r$. It holds that $Q_{p_r,k}\geq B_k(\bl a^r)$.
\end{proposition}

We now bound the utility the improving author obtains in the corresponding improvement step, when her document's quality does not exceed the highest quality (on that particular topic) in the preceding improvement step. 
\begin{proposition}
\label{prop:boundonutility}
Let $\gamma$ be a finite improvement path, and let $a^{r+1}_{p_r}=k$ for an arbitrary improvement step $r$. If  $Q_{p_r,k}\leq B_k(\bl a^r)$, then 
\[
u^{Ex}_{p_{r}}(\bl a^ {r+1})\leq \frac{D(k)}{W_{k}(\gamma)+1}.
\]
\end{proposition}

Next, we characterize a property that must hold in improvement cycles, under the false assumption that such exist. We prove that if an improvement cycle exists, the quality of the first-ranked document is constant throughout the improvement cycle; this must hold for every topic.

\begin{lemma}
\label{lemma:cyclesexpuretargeted}
If $c=(\bl a^1,\ldots,\bl a^l=\bl a^1)$ is an improvement cycle, then for every improvement step $r$ and every topic $k$ it holds that $B_k(\bl a^r)= B_k(\bl a^{r+1})$.
\end{lemma}

\begin{proof}[Proof sketch]
We give here a high-level overview of the proof and refer the reader to the appendix for the formal proof.

Under the false assumption that an improvement cycle exists, assume that the claim does not hold. Namely, assume that $c=(\bl a^1,\ldots,\bl a^l=\bl a^1)$ is an improvement cycle (w.l.o.g. $c$ is a simple improvement cycle), and that there exist an improvement step $r$ and a topic $k$ such that $B_k(\bl a^r)\neq B_k(\bl a^{r+1})$.% We can assume w.l.o.g. that $B_k(\bl a^r)> B_k(\bl a^{r+1})$ (since $c$ is a cycle).

Recall that $D(1)\geq \cdots \geq D(m)$, i.e., the topics are sorted according to the query distribution mass in a non-increasing order. We prove by induction on the topic index $k$ that $B_k(\bl a^r)\leq B_k(\bl a^{r+1})$ holds for every $r$, $1\leq r \leq l-1$. Clearly, if this holds for every improvement step $r$ then
\[
B_k(\bl a^1)\leq  \cdots \leq B_k(\bl a^l)=B_k(\bl a^1);
\]
thus, all inequalities hold in equality and $B_k(\bl a^r)\neq  B_k(\bl a^{r+1})$ cannot occur.

\noindent\textbf{Base, $k=1$:} 
Assume the assertion does not hold for $k=1$; hence, there exists $r$, $1\leq r \leq l-1$, such that $B_1(\bl a^r)>B_1(\bl a^{r+1})$. This means that there exists an author who writes with the highest quality on topic 1 in the step $r$, and then she deviates to another topic in step $r+1$. Moreover, due to the strict inequality, that author is the unique author to write with the highest quality on topic 1 in step $r$; hence, her utility in step $r$ is exactly $D(1)$. When she deviates, she can obtain at most $D(2)$, but recall that $D(1)\geq D(2)$; hence, this deviation is not beneficial.

\noindent\textbf{Step:} Assume the assertion holds for $k\in \{ 1,2,\dots K-1 \}$, i.e., $B_k(\bl a^r)=B_k(\bl a^{r+1})$ for every step $r$.
We show that $B_K(\bl a^r)> B_K(\bl a^{r+1})$ for a step $r$ implies that the improving author in improvement step $r$ deviates to a topic with a lower index. Using the bound obtained in Proposition \ref{prop:boundonutility} and the induction hypothesis, we show that there must be an improving author which does not increase her utility after preforming the deviation, which is clearly a contradiction.
\end{proof}

Lemma \ref{lemma:cyclesexpuretargeted} implies that the only element that varies throughout an improvement cycle, if such exists, is the number of authors who write on each topic. In particular, the highest quality on each topic remains constant. It also suggests that any improving author is not the only author writing the highest quality document on the topic to which she deviated. 

Consider an arbitrary improvement step, and denote by $k$ the topic that the improving author writes on in the improvement step. The improving author joins a (non-empty) set of authors which are already writing documents with the highest quality  on topic $k$. Since we deal with a cycle, at some point an author abandons topic $k$, and deviates to another topic, say $k'$. In Lemma \ref{lemma:changeinqualitsexp} we bound the utility of the improving author (deviating to topic $k$) with that of the author who deviated to $k'$.

\begin{lemma}
\label{lemma:changeinqualitsexp}
If $c=(\bl a^1,\ldots,\bl a^l=\bl a^1)$ is an improvement cycle, then for every improvement step $r$ and topic $k$ such that $a_{p_r}^{r+1}=k$ there exist $(r',k')$ such that $a_{p_{r'}}^{r'+1}=k'$ and
\[
\frac{D(k)}{W_k(c)+1}<\frac{D(k')}{W_{k'}(c)+1}.
\]
\end{lemma}
\begin{proof}[Proof sketch] 
Let $r,k$ be such that $a^{r+1}_{p_r}=k$. By definition of improvement step  $a_{p_r}^{r}\neq k$. From Lemma \ref{lemma:cyclesexpuretargeted} we know that $B_k(\bl a^r)=B_k(\bl a^{r+1})$; thus, $Q_{p_r,k} = B_k(\bl a^r)$ and $H_k(\bl a^r)\neq H_k(\bl a^{r+1})$. Afterwards, we prove another claim which guarantees that there exists $r'$ such that
\begin{equation*}
\frac{D(k)}{W_{k}(c)+1}=u^{Ex}_{p_{r'}}(\bl a^ {r'})
\end{equation*}
holds. In addition, $p_{r'}$ is the improving author, and so
% $u^{Ex}_{p_{r'}}(\bl a^{r'})<u^{Ex}_{p_{r'}}(\bl a^{r'+1})$, Equation (\ref{Lemma2add:0}) suggests that
\begin{equation}
\begin{split}\label{eq:lemmaproofsketch}
\frac{D(k)}{W_{k}(c)+1}=u^{Ex}_{p_{r'}}(\bl a^ {r'})<u^{Ex}_{p_{r'}}(\bl a^ {r'+1}).
\end{split}
\end{equation}
Clearly, $ a^{r'+1}_{p_{r'}}=k'\neq k$. Lemma \ref{lemma:cyclesexpuretargeted} indicates that $B_{k'}(\bl a^{r'})=B_{k'}(\bl a^{r'+1})$; hence, $Q_{p_{r'},k'}\leq B_{k'}(\bl a^{r'})$. Having showed that the condition of Proposition \ref{prop:boundonutility} holds, we invoke it for $r',k'$ and conclude that
\[
u^{Ex}_{p_{r'}}(\bl a^ {r'+1})\leq  \frac{D(k')}{W_{k'}(c)+1}.
\]
Combining this fact with Equation (\ref{eq:lemmaproofsketch}), we get
\[
\frac{D(k)}{W_k(c)+1}<\frac{D(k')}{W_{k'}(c)+1}.
\]
\end{proof} 

In Theorem \ref{thm:exposure-targeted} below we leverage Lemma \ref{lemma:changeinqualitsexp} to show that improvement cycles cannot exist.
\begin{theorem}
\label{thm:exposure-targeted}
$R^{PRP}$ is $u^{Ex}$-learnable.
\end{theorem}

\begin{proof}[Proof of Theorem \ref{thm:exposure-targeted}]
To show that $R^{PRP}$ is $u^{Ex}$-learnable it suffices to show that every improvement path is finite. Moreover, every improvement path cannot contain more than a finite number of different strategy profiles, as $m$ and $n$ are finite; therefore, if $\gamma$ is infinite it must contain an improvement cycle. We are left to prove that $\gamma$ cannot contain an improvement cycle.

Assume by contradiction that $\gamma$ contains an improvement cycle $c=(\bl a^1,\bl a^2,\ldots,\bl a^l=\bl a^1)$. Let $r_1$ be an arbitrary improvement step and denote by $k_1$ the topic such that $a_{p_{r_1}}^{r_1+1}=k_1$.
From Lemma \ref{lemma:changeinqualitsexp} we know that there exist  $(r_2,k_2)$ such that $a^{r_2+1}_{p{r_2}}=k_2$ and
\[
\frac{D(k_1)}{W_{k_1}(c)+1}<\frac{D(k_2)}{W_{k_2}(c)+1}.
\]
Since $a^{r_2+1}_{p_{r_2}}=k_2$, we can now use Lemma \ref{lemma:changeinqualitsexp} again in order to find $(r_3,k_3)$ such that $a^{r_3+1}_{p_{r_3}}=k_3$ and
\[
\frac{D(k_2)}{W_{k_2}(c)+1}<\frac{D(k_3)}{W_{k_3}(c)+1}.
\]
This process can be extended to achieve additional $k_4,k_5,\ldots, k_{m+1}$ such that
\[
\frac{D(k_1)}{W_{k_1}(c)+1}<\frac{D(k_2)}{W_{k_2}(c)+1}
<\ldots<\frac{D(k_{m+1})}{W_{k_{m+1}}(c)+1}.
\]
Since there are only $m$ topics and that the inequality above contains $m+1$ elements, there are at least two elements which are identical; thus we obtain a contradiction. We deduce that an improvement cycle cannot exist. 
\end{proof}

Theorem \ref{thm:exposure-targeted} concludes the analysis of the exposure-targeted utility function.

\subsection{Action-Targeted Utility}
\label{subsec:Impression-Focused}
After analyzing games with exposure-targeted utility, we proceed to action-targeted utility. The main result of this subsection is that $R^{PRP}$ is $u^{Ac}$-learnable, which is analogous to the main result of the previous one. Interestingly, achieving this result requires a more subtle treatment. To motivate it, consider the following: under $u^{Ex}$, in a case where the quality of an author's document on topic $k$ exceeds the quality of all other authors writing on topic $k$, she will not deviate to a topic with a higher index (a topic with a lower or equal user mass). This, however, is not true for $u^{Ac}$. For instance, consider the strategy profile $(2,1)$ in the example given in Subsection \ref{subsec:example}. Under $u^{Ex}$, author 1 cannot increase her utility by deviating to topic 3 (a topic with a lower user mass). In contrast, under $u^{Ac}$, author 1 \textit{can} improve her utility by deviating to topic 3. To assist in that, let $S_k(\gamma)$ denote the highest quality of a document written on topic $k$ throughout a finite improvement path $\gamma$. Formally, 
\begin{definition}
Given a topic $k$ and an improvement path $\gamma=(\bl a^1,\dots ,\bl a^l)$,
\[
S_k(\gamma)\defeq\max\limits_{1 \leq r \leq l}\{B_k(\bl a^r)\}.
\] 
\end{definition}
In Proposition \ref{prop:boundonutilityUAC} we bound the utility of an improving author in an improvement step. 

\begin{proposition}
\label{prop:boundonutilityUAC}
Let $\gamma$ be a finite improvement path, and let $a^{r+1}_{p_r}=k$ for an arbitrary improvement step $r$. If $Q_{p_r,k}\leq B_k(\bl a^r)$, then 
\[
u^{Ac}_{p_{r}}(\bl a^ {r+1})\leq \frac{D(k)\cdot S_k(\gamma)}{W_{k}(\gamma)+1}.
\]
\end{proposition}

Notice that $S_k(\gamma) \leq 1$ for every $k$ and every $\gamma$; thus, the bound given in Proposition \ref{prop:boundonutility} trivially holds for $u^{Ac}$. However, proving this tighter bound becomes essential for refuting the existence of improvement cycles under $u^{Ac}$. By proving additional supporting lemmas (which are further elaborated in the appendix), we show that 
\begin{theorem}
\label{thm:action-targeted}
$R^{PRP}$ is $u^{Ac}$-learnable.
\end{theorem}

\section{Non-Learnability under Other Mediators}
\label{sec:nonlearnability}
In the previous section we showed a powerful result: 
$R^{PRP}$ is both $u^{Ex}$-learnable and $u^{Ac}$-learnable. In other words,
when using $R^{PRP}$, any better-response dynamics converges; this is true for both utility schemes. In fact, $R^{PRP}$ is not the only mediator under which such convergence occurs. For instance, Let  $R^{RAND}$ be the \textit{random} mediator, such that for any author $j$ and any topic $k$,
\[
R_j^{RAND}(Q,k,\bl a) \defeq 
\begin{cases}
\frac{1}{\sum_{i=1}^n\ind_{a_i=k}} & a_j = k\\
0 & \text{otherwise}
\end{cases}.
\]
By showing that under $u^{Ex}$ any game with $R^{RAND}$ can be reduced to a game with $R^{PRP}$, we conclude that
\begin{proposition}
\label{prop:randislearnable}
$R^{RAND}$ is $u^{Ex}$-learnable.
\end{proposition}
\begin{proof}[Proof sketch]
We prove the claim by showing that under $u^{Ex}$ any game with $R^{RAND}$ can be reduced to a game with $R^{PRP}$, such that the two games are strategically equivalent.
This is done by taking any game $G$ with $R^{RAND}$ as the mediator and a quality matrix $Q$, and reduce it to a game $G'$ with $R^{PRP}$ as the mediator and $Q'$ as the quality matrix, such that $Q_{j,k}' = 1$ for every $ j \in N$ and $k \in M$ . 

Since both $G,G'$ consists of the exposure-targeted utility function, we omit the super-script $Ex$ and use the super-script $G$ to specify the utility of author $j$ under the strategy profile $\bl a$ in $G$, i.e., $u^G_j(\bl a)$, and equivalently $u^{G'}_j(\bl a)$ for $G'$. By definition of exposure-targeted utility and $R^{PRP}$, for every valid $j$ and $\bl a$ it holds that
\begin{align*}
u^{G'}_j(\bl a)&= \sum_{k=1}^m \ind_{a_j=k} \cdot D(k)\cdot R^{PRP}_j(Q',k,\bl a) \\ &=D(a_j)\cdot R^{PRP}_j(Q',a_j,\bl a)\\
&= D(a_j)\cdot \frac{1}{H_{a_j}(\bl a)}\\
&=D(a_j)\cdot R^{RAND}_j(Q,a_j,\bl a)\\
&= \sum_{k=1}^m \ind_{a_j=k} \cdot D(k)\cdot R^{RAND}_j(Q,k,\bl a) =u^{G}_j(\bl a).
\end{align*}
Since $G'$ possesses $R^{PRP}$ as the mediator, Theorem \ref{thm:exposure-targeted} guarantees that $G'$ has the FIP property. Since we showed $G$ and $G'$ are strategically equivalent, $G$ also has the FIP property, and in particular does not contain improvement cycles.
\end{proof}

Notice that $R^{RAND}$  treats every document the same, regardless of its quality. However, in many (and perhaps even most) scenarios mediators seek to promote high-quality content. Therefore, the reader may wonder whether other plausible mediators are  $u^{Ex}$-learnable or $u^{Ac}$-learnable. We now focus on a wide and intuitive family of mediators, which we term scoring mediators.
\begin{definition}
Let $R$ be a mediator. We say that $R$ is a scoring mediator if there exists a non-decreasing function $f:\mathbb R \rightarrow \mathbb R_{+}$ such that for every $Q, k, \bl a$ and author index $j$ it holds that
\[
R_j(Q, k, \bl a) \defeq
\begin{cases}
\frac{f(Q_{j,k})}{\sum_{i=1}^{n} \ind_{a_i = k}\cdot f(Q_{i,k})} & a_j = k\\
0 & \text{otherwise}
\end{cases}.
\]
It this case, we denote $R=R^f$ for the corresponding $f$.
\end{definition}
Under a scoring mediator every author receives a probability according to the proportion of her score over the sum of the scores of all author writing on that topic. 
Notice that if $R^f$ is a scoring mediator such that the corresponding $f$ is constant, we get $R^f=R^{RAND}$. In addition, this family also includes celebrated mediators, e.g. the softmax function (for $f(Q_{j,k})=e^{Q_{j,k}}$), which is very popular in machine learning applications, or the linear function (for $f(Q_{j,k})=Q_{j,k}$) that is common in probabilistic models for decision making (for instance, in the Bradley–-Terry model \cite{bradley1952rank}). Noticeably, the $R^{PRP}$ is not a scoring mediator. 
In the rest of this section, we show non-convergence of better-response dynamics for general families of scoring mediators.

\subsection{Exposure-Targeted Utility}
\label{subsec:Exposure-Targeted-Utility Non Leranbility}
In this subsection we prove that, under mild assumptions, scoring mediators are not $u^{Ex}$-learnable (as opposed to $R^{RAND}$). We restrict ourselves to mediators for which the corresponding function $f$ is continuous, and exhibits the following property: the ratio between the score of the highest quality and the lowest quality is greater than two (note that this property holds trivially if the score of the lowest quality is zero, i.e., $f(0)=0$). Among others, this class of mediators contains mediators based on softmax and linear functions, as described above.  
\begin{theorem}
\label{thm:exposure-targeted NL proof}
Let $R^f$ be a scoring mediator. If $f$ is a continuous function such that $f(1)>2 f(0)$, then $R^f$ is not $u^{Ex}$-learnable.
\end{theorem}
\begin{proof}[Proof sketch]
It is sufficient to show that for every $f$ that satisfies the theorem's conditions, we can construct a game instance with an improvement cycle. We exploit the properties of $f$ to construct a game with four authors and three topics, and show that an improvement cycle exists. Let $R^f$ be a scoring mediator with the corresponding function $f$, which we assume exhibits $f(1)> 2f(0)$. Due to the Intermediate Value Theorem, there exist $x_1,x_2,x_3$ such that $0<x_3<x_2<x_1\leq 1$ and 
\[
\frac{f(x_2)}{f(x_3)}>\frac{2f(x_1)}{f(x_2)}>2.
\]
For brevity, denote $c_1=\frac{f(x_1)}{f(x_2)}$ and $c_2=\frac{f(x_2)}{f(x_3)}$, and observe that $c_2>2 c_1$. Consider a game with $\abs{N}=4$ authors, $\abs{M}=3$ topics and a quality matrix $Q$ such that 
\[
\begin{pmatrix}
x_1 & 0 & 0 \\
x_1 & x_2 & 0 \\
x_2 & 0 & x_3\\
0 & x_3 & x_2
\end{pmatrix}.
\]
The only missing ingredient is the distribution $D$ over the topics. The selection of such $D$ is crucial: we shall select $D$ to allow improvement cycles. Denote
\[
D(1)=\frac{1}{2-3\epsilon},  D(2)=\frac{1-2\epsilon}{2(2-3\epsilon)} ,  D(3)=\frac{1-4\epsilon}{2(2-3 \epsilon)},
\]
for some $0<\epsilon\leq \frac{1}{4}$. It can be verified that $D$ is a valid distribution over the set of topics.  Consider the strategy profiles
\begin{gather*}
\bl a^1=(1,1,1,2), \quad
\bl a^2=(1,2,1,2), \quad
\bl a^3=(1,2,3,2), \\
\bl a^4=(1,2,3,3), \quad
\bl a^5=(1,1,3,3), \quad
\bl a^6=(1,1,1,3).
\end{gather*}
In the rest of the proof we show that $\epsilon$ can be selected such that the cycle $c=(\bl a^1,\bl a^2,\bl a^3, \bl a^4,\bl a^5,\bl a^6, \bl a^7=\bl a^1)$
is an improvement cycle of the game we constructed. More precisely, we prove that for every $r$, $1\leq r \leq 6$, $u^{Ex}_{p_r}(\bl a^r)<u^{Ex}_{p_r}(\bl a^{r+1})$. This suggests that $R^f$ is not $u^{Ex}$-learnable.
\end{proof}
While all it takes to prove Theorem \ref{thm:exposure-targeted NL proof} is to show a \textit{single} game instance with an improvement cycle, we can actually construct infinitely many games which do not possess FIP. Moreover, our construction can be viewed as a sub-game in a much broader game, i.e., with more authors and topics.

\subsection{Action-Targeted Utility} 
\label{subsec:Action-Targeted-Utility Non Leranbility}
When analyzing scoring mediators, an additional difference between the two utility schemes emerges. In the improvement cycle constructed in the proof of Theorem \ref{thm:exposure-targeted NL proof}, there exists an improvement step in which the improving author decreases the quality of her document but still increases her utility. Under the action-targeted utility function, such a decrease may not be translated to improved utility. Namely, the technique employed in Theorem \ref{thm:exposure-targeted NL proof} for constructing a game that possesses an improvement cycle might not work here. Nevertheless, the following theorem shows non-learnability under $u^{Ac}$ of scoring mediators that boost high-quality content. For example, a mediator $R^f$ where the corresponding $f$ satisfies $f(1)>6f(\frac{1}{2})$ assigns a substantially higher score to the highest quality than a mediocre one.
\begin{theorem}
\label{thm:action-targeted NL proof11}
Let $R^f$ be a scoring mediator. If $f$ is a continuous function such that  ${f(1)>2(2\alpha -1) f\left(\frac{1}{\alpha}\right)}$
for some $\alpha >1$, then $R^f$ is not $u^{Ac}$-learnable. 
\end{theorem}
Notice the resemblance between the condition of Theorem \ref{thm:exposure-targeted NL proof} to that of Theorem \ref{thm:action-targeted NL proof11}. Due to space limitations, additional results on the non-learnability of other scoring mediators under $u^{Ac}$ are omitted and further elaborated in the appendix.
\section{Discussion} 
\label{sec:disc}
We introduced the study of learning dynamics in the context of information retrieval games. Our results address learning in the framework  introduced by \citeauthor{basat2017game}  \shortcite{basat2017game}, where authors are action-targeted as well as for a complementary type of information retrieval game in which the authors' aim is to maximize their exposure. In particular, our results show that a mediator that operates according to the PRP \cite{Robertson:77a} induces a game in which learning-dynamics converges; the latter is true for both exposure-targeted and action-targeted utility schemes. Moreover, we have also demonstrated that this convergence is a virtue of the PRP, and does not apply for other relevant mediators. 

One prominent question is the time required for the authors to converge, namely, finding the worst-case length of an improvement path. It turns out that there is a class of games where the length of the best-response paths is easy to analyze.

Consider the exposure target utility, and assume that $D$ is strictly decreasing, the number of authors equals the number of topics, and that the matrix $Q$ is generic, i.e., has $n \times m$ distinct values. The induced game exhibits a unique equilibrium:
topic 1 is assigned to the author with the highest quality w.r.t. topic 1. Topic 2 is assigned to the author with the highest quality on that topic, from the set of authors who were not assigned before. Clearly, the PNE is computed by following this process until every author/topic is assigned. Consequently, any best-response dynamics where the authors play in a round-robin fashion will converge after at most a quadratic number of improvement steps in the number of authors. A similar observation applies to action targeted utility under a slightly different notion of generality of $Q$. The general question of convergence rate is nevertheless left open.

Our model, as any other novel model that pretends to explain theoretical aspects of real-world systems, has its limitations. To name a few, we assume the set of authors and topics are fixed, while in reality they are often dynamic; we assume that the quality of documents is perfectly observed by the mediator, which only approximates modern search engines. Although not ultimate, we do believe that our modeling, which extends a model that is already acknowledged as valuable \cite{basat2017game}, serves as an important justification for the use of the PRP, and may be an important step for future work to circumvent the limitations presented above. We note that our learning dynamics is based on applying an author's response to the current behavior of other authors. In fact, this assumes that the only information available to the author is the quality of the documents currently published, and assumes nothing about information available to an author on other authors' (unobserved) qualities. Relaxing the assumption that published documents' qualities can be observed goes beyond the scope of our work, and may be a subject for future research.

An interesting future direction is to expand the information retrieval setting to a setup where each author's document may include several topics. This issue is treated in a preliminary manner in \cite{basat2017game} and it may be of interest to see whether our results can be extended to that context as well. It may be also interesting to study the quality of the equilibrium (as far as users' social welfare is concerned) reached under PRP. Would the best equilibrium  be obtained under better-response learning dynamics? 

\section*{Acknowledgments}\label{sec:Acknowledgments}
This project has received funding from the European Research Council (ERC) under the European Union's Horizon 2020 research and innovation programme (grant agreement n$\degree$  740435).

{
\small
\bibliographystyle{aaai}
%\bibliography{References}

}
{\ifnum\Includeappendix=1{ 

\appendix
\onecolumn

\section{Omitted Proofs from Section \ref{sec:PRP}}
\begin{proof}[\textbf{Proof of Proposition \ref{prop:pot}}]
The proof makes use of the following theorems by \citeauthor{MondererShapley96}, which are stated slightly different for ease of presentation.

\begin{theorem}[Theorem 2.8,\cite{MondererShapley96}]
\label{thm:potential}
Let $G$ be a two-by-two bimatrix game such that 
\[
\begin{pmatrix}
      a_{1,1},b_{1,1} & a_{1,2},b_{1,2}   \\
      a_{2,1},b_{2,1}&  a_{2,2},b_{2,2}  \\
\end{pmatrix}.
\]
Then $G$ is an exact potential game if and only if 
\begin{equation}
\label{eq:haspotential}
a_{2,1}-a_{1,1}+b_{2,2}-b_{2,1}+a_{1,2}-a_{2,2}+b_{1,1}-b_{1,2}=0
\end{equation}
\end{theorem}
Further,
\begin{theorem}[Corollary 2.9,\cite{MondererShapley96}]
\label{moshapely}
A game $G$ is an exact potential game if and only if every two-by-two subgame of $G$ is an exact potential game.
\end{theorem}

To show that the class of games induced by the PRP mediator and $u^{Ex}$ (equivalently, $u^{Ac}$) does not have an exact potential, it is sufficient to show a subgame of a larger game which is not an exact potential game. 

Consider a game $G$ with $n=3$ authors, $m=2$ topics, $D(1)=D(2)=0.5$ and a quality matrix
\[
Q=
\begin{pmatrix}

0.3 & 0.4 \\
0.5 & 0.7 \\
0.1 & 0.4  \\
\end{pmatrix}.
\]
We first focus on $u^{Ex}$. Observe that the bimatrix game describing the utilities of authors 1 (rows) and 2 (columns) induced by setting $a_3=2$ is
\[
\kbordermatrix{
     & \text{topic } 1 & \text{topic } 2 \\
     \text{topic } 1 & 0,0.5 & 0.5,0.5   \\
     \text{topic } 2 & 0.25,0.5 & 0,0.5 \\
     }.
\]
The above bimatrix game does not satisfy the condition given in Equality (\ref{eq:haspotential}), since
\[
0.25-0+0.5-0.5+0.5-0+0.5-0.5=0.75\neq 0;
\]
hence, Theorem \ref{thm:potential} implies that $G$ does not have an exact potential.

Next, consider a game $G'$ with the same $N,M,Q,D$, and let $u^{Ac}$ be the utility function. Let $a_3=2$, and observe that the utility bimatrix of authors 1 and 2 is
\[
\kbordermatrix{
     & \text{topic } 1 & \text{topic } 2 \\
     \text{topic } 1 & 0,0.25 & 0.15,0.35   \\
     \text{topic } 2 & 0.1,0.25 & 0,0.35 \\
     }. 
\]
Here again,
\[
0.1-0+0.35-0.25+0.15-0+0.25-0.35=0.25\neq 0;
\]
thus, Theorem \ref{thm:potential} implies that $G'$ does not have an exact potential.
\end{proof}

\section{Omitted Proofs from Subsection \ref{subsec:Exposure-Targeted-Utility}}

\begin{proof}[\textbf{Proof of Proposition \ref{prop:leqgeq}}]
Since author $p_r$ improves her utility, $u^{Ex}_{p_{r}}(\bl a^{r})<u^{Ex}_{p_r}(\bl a^{r+1})$. By definition of $R^{PRP}$, if $Q_{p_r,k} < B_{k}(\bl a^{r})$ then $u^{Ex}_{p_r}(\bl a^{r+1})=0\leq u^{Ex}_{p_{r}}(\bl a^{r})$, which results in a contradiction.
\end{proof}

\begin{proof}[\textbf{Proof of Proposition \ref{prop:boundonutility}}]
Combined with Proposition \ref{prop:leqgeq}, we know that 
\begin{eqnarray}
\label{prop:dmflssdkmf}
Q_{p_{r},k} = B_{k}(\bl a^{r}).
\end{eqnarray}
Notice that $a_{p_r}^{r} \neq k$ and $a_{p_r}^{r+1} = k$; hence, together with Equation (\ref{prop:dmflssdkmf}) we obtain
\begin{equation}
\label{prop:eq:dasasdfk}
H_{k}(\bl a^{r+1})=H_{k}(\bl a^{r})+1\stackrel{\text{Def. of }W_{k}(\gamma) }{\geq} W_{k}(\gamma)+1.
\end{equation}
Observe that Equation (\ref{prop:eq:dasasdfk}) suggests that
\[
u^{Ex}_{p_{r}}(\bl a^ {r+1})=\frac{D(k)}{H_{k}(\bl a^{r+1})} =\frac{D(k)}{H_{k}(\bl a^{r})+1}\leq \frac{D(k)}{W_{k}(\gamma)+1},
\]
which concludes the proof of this proposition.
\end{proof}

\begin{proposition}
\label{prop:boundonutilityequal}
If $c=(\bl a^1,\dots,\bl a^l=\bl a^1)$ is an improvement cycle and $k$ is a topic such that the following properties hold
\begin{enumerate}
\item \label{assumptionone} there exists an improvement step $r_1$ satisfying $H_k(\bl a^{r_1})\neq H_k(\bl a^{r_1+1})$, and
\item \label{assumptiontwo} for every improvement step $r_2$, $B_k(\bl a^{r_2})=B_k(\bl a^{r_2+1})$,
\end{enumerate}
then there exist an index $r$ such that $a_{p_{r}}^{r}=k$ and
\[
u^{Ex}_{p_{r}}(\bl a^ {r})= \frac{D(k)}{W_{k}(c)+1}.
\]
\end{proposition}

\begin{proof}[\textbf{Proof of Proposition \ref{prop:boundonutilityequal}}]
From Property \ref{assumptionone} we know that there exists an improvement step $r_1$ such that $H_k(\bl a^{r_1})\neq H_k(\bl a^{r_1+1}).$ Assume w.l.o.g. that $H_k(\bl a^{r_1})> H_k(\bl a^{r_1+1})$; hence
\begin{equation}
\label{prop4:1}
H_k(\bl a^{r_1})> H_k(\bl a^{r_1+1})\geq W_k(c)+1.
\end{equation}
By the definition of $W_k(c)$ we know that there exists an improvement step  $r_3$ such that \begin{equation}
\label{prop4:2}
H_k(\bl a^{r_3})=W_k(c).
\end{equation}
From Property \ref{assumptiontwo} we get that for every improvement step $r_2$, $B_k(\bl a^{r_2})=B_k(\bl a^{r_2+1})$, which implies that
\begin{equation}
\label{prop4:3}
\abs{H_k(\bl a^{r_2})-H_k(\bl a^{r_2+1})}\leq 1
\end{equation}
Combining Equations (\ref{prop4:1}),(\ref{prop4:2}) and (\ref{prop4:3}) with the fact that $c$ is an improvement cycle leads to the fact that there must exist an improvement step $r$ such that $\bl a^{r} \in \{\bl a^{r_1},\bl a^{r_1+1},\ldots,\bl a^{r_3-1}\}$, $H_k(\bl a^{r})=W_k(c)+1$ and $H_k(\bl a^{r+1})=W_k(c)$. This implies that $a_{p_{r}}^{r}=k$ and $Q_{p_{r},k}=B_k(\bl a^{r})$; therefore,

\begin{equation*}
u^{Ex}_{p_{r}}(\bl a^{r})=\frac{D(k)}{H_{k}(\bl a^{r})}=\frac{D(k)}{W_{k}(c)+1}.
\end{equation*}
\end{proof}

\begin{proof}[\textbf{Proof of Lemma \ref{lemma:cyclesexpuretargeted}}]

Assume w.l.o.g. that $c$ is a simple improvement cycle. First, we prove by induction on the topic index $k$ that $B_k(\bl a^r)\leq B_k(\bl a^{r+1})$ holds for every $r$, $1\leq r \leq l-1$. Later, we leverage this result to prove the statement of the lemma.

\textbf{Base:} 
Assume the assertion does not hold for $k=1$; hence, there exists $r$, $1\leq r \leq l-1$, such that 
$B_1(\bl a^r)>B_1(\bl a^{r+1})$. As a result, it holds for the improving author $p_r$ in step $r$ that $Q_{p_r,1} > B_1(\bl 
a_{-p_r}^r)$ and $H_1(\bl a^r)=1$. In words, the quality of $p_r$'s document exceeds all other qualities under $\bl a^r$ on topic 1; thus,
\begin{equation}
\begin{split}\label{Lemma1:Add0}
D(1)=u^{Ex}_{p_r}(\bl a^{r}).
\end{split}
\end{equation}
In addition, $\bl a^{r+1}$ is an improvement step for author $p_r$, and so $u^{Ex}_{p_r}(\bl a^{r})<u^{Ex}_{p_r}(\bl 
a^{r+1})$. Combined with Equation (\ref{Lemma1:Add0}), 
\begin{equation}
\begin{split}\label{Lemma1:Add1}
D(1)<u^{Ex}_{p_r}(\bl a^{r+1}).
\end{split}
\end{equation}
On the other hand, $u^{Ex}_{p_r}(\bl a^{r+1})\leq D(a_{p_r}^{r+1})$ holds; thus, Equation (\ref{Lemma1:Add1}) implies that $D(1)< D(a_{p_r}^{r+1})$, which is clearly a contradiction since $D(1)\geq \ldots \geq D(m)$.

\textbf{Step:}  Suppose the assertion holds for every $k$ where $k<K\leq m$, but does not hold for $K$. Similarly to the base case, there exists $r_1$, $1\leq r_1 \leq l-1$, such that $B_K(\bl a^{r_1})>B_K(\bl a^{r_1+1})$. As a result, $Q_{p_{r_1},K} > B_K(\bl a_{-p_{r_1}}^{r_1})$  and $H_K(\bl a^{r_1})=1$ hold, implying that 
\begin{equation}
\begin{split}\label{Lemma1:Add2}
D(K)=u^{Ex}_{p_{r_1}}(\bl a^{r_1}).
\end{split}
\end{equation}
In addition, $u^{Ex}_{p_{r_1}}(\bl a^{r_1})<u^{Ex}_{p_{r_1}}(\bl a^{r_1+1})$ holds since $p_{r_1}$ is the improving author; hence, with Equation (\ref{Lemma1:Add2}) we get
\begin{equation}
\begin{split}\label{Lemma1:Add4}
D(K)<u^{Ex}_{p_{r_1}}(\bl a^ {r_1+1}).
\end{split}
\end{equation}
Let $k_1$ denote the topic that author $p_{r_1}$ is writing on under $\bl a^{r_1+1}$, i.e., $k_1=a^{r_1+1}_{p_{r_1}}$. By definition of $u^{Ex}$ we obtain  
\begin{equation}
\begin{split}\label{Lemma1:Add5}
u^{Ex}_{p_{r_1}}(\bl a^{r_1+1})\leq D(k_1).
\end{split}
\end{equation}
Recall that $D(1)\geq \dots \geq D(m)$; hence, Equations (\ref{Lemma1:Add4}) and (\ref{Lemma1:Add5}) suggest that $D(K)<D(k_1)$ holds, and therefore we are guaranteed that $k_1<K$.
%From Equations (\ref{Lemma1:Add4}) and (\ref{Lemma1:Add5}) we get that $D(K)<D(k_1)$ which together with our assumption of $D(1)\geq \dots \geq D(m)$ implies that $k_1<K$.

Since $k_1<K$, the induction hypothesis hints that  $B_{k_1}(\bl a^{r_1})=B_{k_1}(\bl a^{r_1+1})$; therefore, $Q_{p_{r_1},k_1}\leq B_{k_1}(\bl a^{r_1})$ holds and by Proposition \ref{prop:leqgeq} we get that $Q_{p_{r_1},k_1} = B_{k_1}(\bl a^{r_1})$. Notice that $c$ is a finite improvement path, and that the condition of Proposition \ref{prop:boundonutility} holds; thus, by invoking it for $r_1,k_1$, we get
\[
u^{Ex}_{p_{r_1}}(\bl a^ {r_1+1})\leq \frac{D(k_1)}{W_{k_1}(c)+1}.
\]

Together with Equation (\ref{Lemma1:Add4}), we conclude that
\begin{equation}
\label{Lemma1:Kk1}
D(K)<\frac{D(k_1)}{W_{k_1}(c)+1}.
\end{equation}

Next, we wish to find an improvement step such that the improving author's utility strictly bounds the right-hand-side of Equation (\ref{Lemma1:Kk1}). Since $a_{p_{r_1}}^{r_1+1}=k_1$ and $Q_{p_{r_1},k_1} = B_{k_1}(\bl a^{r_1})$ we get that $H_{k_1}(\bl a^{r_1}) \neq H_{k_1}(\bl a^{r_1+1})$. In addition, from the induction hypothesis, we get that for every improvement step $r'$, $B_{k_1}(\bl a^{r'})=B_{k_1}(\bl a^{r'+1})$; hence we can invoke Proposition \ref{prop:boundonutilityequal} which guarantees the existence of an index $r_2$ such that $a_{p_{r_2}}^{r_2}=k_1$ and
\begin{equation}
\begin{split}\label{Lemma1:3}
\frac{D(k_1)}{W_{k_1}(c)+1}=u^{Ex}_{p_{r_2}}(\bl a^{r_2}).
\end{split}
\end{equation}
Since $p_{r_2}$ is the improving author $u^{Ex}_{p_{r_2}}(\bl a^{r_2})<u^{Ex}_{p_{r_2}}(\bl a^{r_2+1})$ holds, which together with Equation (\ref{Lemma1:3}) implies
\begin{equation}
\begin{split}\label{Lemma1:Add10}
\frac{D(k_1)}{W_{k_1}(c)+1}<u^{Ex}_{p_{r_2}}(\bl a^ {r_2+1}).
\end{split}
\end{equation}
Let $a^{r_2+1}_{p_{r_2}}=k_2$. By definition of $u^{Ex}$, we know that
\begin{equation}
\begin{split}\label{Lemma1:Add8}
u^{Ex}_{p_{r_2}}(\bl a^{r_2+1})\leq D(k_2).
\end{split}
\end{equation}
Observe that $k_2<K$ must hold. To see this, assume otherwise that $k_2\geq K$, and $D(k_2)\leq D(K)$ follows. Incorporating this assumption with Equations (\ref{Lemma1:Kk1}),(\ref{Lemma1:Add10}) and (\ref{Lemma1:Add8}) we obtain
\[
D(K)<\frac{D(k_1)}{W_{k_1}(c)+1}<u^{Ex}_{p_{r_2}}(\bl a^{r_2+1})\leq D(k_2) \leq D(K),
\]
which is a contradiction; hence, $k_2<K$. The induction hypothesis hints that  $B_{k_2}(\bl a^{r_2})=B_{k_2}(\bl a^{r_2+1})$, implying $Q_{p_{r_2},k_2}\leq B_{k_2}(\bl a^{r_2})$.

Here again, the condition of Proposition \ref{prop:boundonutility} holds; thus, by invoking it for $r_2,k_2$ we conclude that
\[
u^{Ex}_{p_{r_2}}(\bl a^ {r_2+1})\leq \frac{D(k_2)}{W_{k_2}(c)+1}.
\]
Together with Equation (\ref{Lemma1:Add10}), we conclude that
\[
\frac{D(k_1)}{W_{k_1}(c)+1}<\frac{D(k_2)}{W_{k_2}(c)+1}.
\]
We have therefore bound the right-hand-side of Equation (\ref{Lemma1:Kk1}) as desired.

This process can be extended to obtain additional $k_3,k_4,\ldots,k_{K}$, such that for all $i\in[K]$, $k_i<K$ and 
\[
\frac{D(k_{1})}{W_{k_{1}}(c)+1}<\frac{D(k_{2})}{W_{k_{2}}(c)+1}<\frac{D(k_{3})}{W_{k_{3}}(c)+1}< \ldots < \frac{D(k_{K})}{W_{k_{K}}(c)+1}.
\]
While the inequality above contains $K$ elements, there are only $K-1$ topics with index lower than $K$; hence, at least two of them must be identical, and we obtain a contradiction. We deduce that $B_K(\bl a^{r})\leq B_K(\bl a^{r+1})$ for every step $r$.

This concludes the proof of the induction. Ultimately, to end the proof of this lemma, fix a topic $k$. Due to the induction above, $B_k(\bl a^r)\leq B_k(\bl a^{r+1})$ holds for every $1\leq r \leq l-1$, i.e.,
\[
B_k(\bl a^1)\leq B_k(\bl a^2)\leq\ldots\leq B_k(\bl a^{l-1})\leq B_k(\bl a^l)=B_k(\bl a^1).
\]
The left-hand-side and the right-hand-side of the inequality above are identical; thus, they must all hold in equality. This concludes the proof of this lemma.
\end{proof}

\begin{proof}[\textbf{Proof of Lemma \ref{lemma:changeinqualitsexp}}]
Let $r,k$ such that $a^{r+1}_{p_r}=k$. From Lemma \ref{lemma:cyclesexpuretargeted} we know that for every improvement step $r''$, $B_k(\bl a^{r''})=B_k(\bl a^{r''+1})$; thus, $Q_{p_r,k}\leq B_k(\bl a^r)$ which by Proposition \ref{prop:leqgeq} leads to
\begin{equation}
\label{newLemma2}
Q_{p_r,k} = B_k(\bl a^r).
\end{equation}
By definition of improvement step  $a_{p_r}^{r}\neq k$; hence together with Equation (\ref{newLemma2}) we get that $H_k(\bl a^{r})\neq H_k(\bl a^{r+1}).$  Notice that $c$ is a finite improvement path, and that the condition of Proposition \ref{prop:boundonutilityequal} holds; hence, by invoking it for $r,k$ we conclude the existence of an index $r'$ such that $a_{p_{r'}}^{r'}=k$ and
\begin{equation*}
\frac{D(k)}{W_{k}(c)+1}=u^{Ex}_{p_{r'}}(\bl a^ {r'}).
\end{equation*}
In addition, $p_{r'}$ is the improving author, and so
\begin{equation}
\begin{split}\label{Lemma2add:1}
\frac{D(k)}{W_{k}(c)+1}=u^{Ex}_{p_{r'}}(\bl a^ {r'})<u^{Ex}_{p_{r'}}(\bl a^ {r'+1}).
\end{split}
\end{equation}
Clearly, $ a^{r'+1}_{p_{r'}}=k'\neq k$. Lemma \ref{lemma:cyclesexpuretargeted} indicates that $B_{k'}(\bl a^{r'})=B_{k'}(\bl a^{r'+1})$; hence, $Q_{p_{r'},k'}\leq B_{k'}(\bl a^{r'})$. Having showed the condition of Proposition \ref{prop:boundonutility} holds, we invoke it for $r',k'$ and conclude that
\[
u^{Ex}_{p_{r'}}(\bl a^ {r'+1})\leq  \frac{D(k')}{W_{k'}(c)+1}.
\]
Combining this fact with Equation (\ref{Lemma2add:1}), we get
\[
\frac{D(k)}{W_k(c)+1}<\frac{D(k')}{W_{k'}(c)+1}.
\]
\end{proof}

\section{Omitted Proofs from Subsection \ref{subsec:Impression-Focused}}

\begin{proposition}
\label{prop:leqgeqUAC}
Let $\gamma$ be a finite improvement path, and let $a^{r+1}_{p_r}=k$ for an arbitrary improvement step $r$. It holds that $Q_{p_r,k}\geq B_k(\bl a^r)$.
\end{proposition}
\begin{proof}[\textbf{Proof of Proposition \ref{prop:leqgeqUAC}}]
Similarly to the proof of Proposition \ref{prop:leqgeq}, since author $p_r$ improves her utility, $u^{Ac}_{p_{r}}(\bl a^{r})<u^{Ac}_{p_r}(\bl a^{r+1})$. By definition of $R^{PRP}$, if $Q_{p_r,k} < B_{k}(\bl a^{r})$ then $u^{Ac}_{p_r}(\bl a^{r+1})=0\leq u^{Ac}_{p_{r}}(\bl a^{r})$, which results in a contradiction.
\end{proof}

\begin{proof}[\textbf{Proof of Proposition \ref{prop:boundonutilityUAC}}]
Combined with Proposition \ref{prop:leqgeqUAC}, we know that 
\begin{eqnarray}
\label{prop:dmflssdkmfUAC}
Q_{p_{r},k} = B_{k}(\bl a^{r}).
\end{eqnarray}
Notice that $a_{p_r}^{r} \neq k$ and $a_{p_r}^{r+1} = k$; hence, together with Equation (\ref{prop:dmflssdkmfUAC}) we obtain
\begin{equation}
\label{prop:eq:dasasdfkUAC}
H_{k}(\bl a^{r+1})=H_{k}(\bl a^{r})+1\stackrel{\text{Def. of }W_{k}(\gamma) }{\geq} W_{k}(\gamma)+1.
\end{equation}
Observe that Equation (\ref{prop:eq:dasasdfkUAC}) suggests that
\[
u^{Ac}_{p_{r}}(\bl a^ {r+1})=\frac{D(k)\cdot Q_{p_r,k}}{H_{k}(\bl a^{r+1})} =\frac{D(k)\cdot B_k(\bl a^r)}{H_{k}(\bl a^{r})+1}\leq \frac{D(k)\cdot S_k(\gamma)}{W_{k}(\gamma)+1},
\]
which concludes the proof of this proposition.
\end{proof}

\begin{proposition}
\label{prop:boundonutilityequalUAC}
If $c=(\bl a^1,\dots,\bl a^l=\bl a^1)$ is an improvement cycle and $k$ is a topic such that the following properties hold
\begin{enumerate}
\item \label{assumptiononeUAC} there exists an improvement step $r_1$ satisfying $H_k(\bl a^{r_1})\neq H_k(\bl a^{r_1+1})$, and
\item \label{assumptiontwoUAC} for every improvement step $r_2$, $B_k(\bl a^{r_2})=S_k(c)$,
\end{enumerate}
then there exist an index $r$ such that $a_{p_{r}}^{r}=k$ and
\[
u^{Ac}_{p_{r}}(\bl a^ {r})= \frac{D(k)\cdot S_k(c)}{W_{k}(c)+1}.
\]
\end{proposition}

\begin{proof}[\textbf{Proof of Proposition \ref{prop:boundonutilityequalUAC}}]
From Property \ref{assumptiononeUAC} we know that there exists an improvement step $r_1$ such that $H_k(a^{r_1})\neq H_k(a^{r_1+1}).$ Assume w.l.o.g. that $H_k(\bl a^{r_1})> H_k(\bl a^{r_1+1})$; hence
\begin{equation}
\label{prop4:1UAC}
H_k(\bl a^{r_1})> H_k(\bl a^{r_1+1})\geq W_k(c)+1.
\end{equation}
By the definition of $W_k(c)$ we know that there exists an improvement step  $r_3$ such that \begin{equation}
\label{prop4:2UAC}
H_k(\bl a^{r_3})=W_k(c).
\end{equation}
From Property \ref{assumptiontwoUAC} we get that for every improvement step $r_2$, $B_k(\bl a^{r_2})=S_k(c)$, which implies that
\begin{equation}
\label{prop4:3UAC}
\abs{H_k(\bl a^{r_2})-H_k(\bl a^{r_2+1})}\leq 1
\end{equation}
Combining Equations (\ref{prop4:1UAC}),(\ref{prop4:2UAC}) and (\ref{prop4:3UAC}) with the fact that $c$ is an improvement cycle leads to the fact that there must exist an improvement step $r$ such that $\bl a^{r} \in \{\bl a^{r_1},\bl a^{r_1+1},\ldots,\bl a^{r_3-1}\}$, $H_k(\bl a^{r})=W_k(c)+1$, and $H_k(\bl a^{r+1})=W_k(c)$. This implies that $a_{p_{r}}^{r}=k$ and $Q_{p_{r},k}=B_k(\bl a^{r})=S_k(c)$; therefore,

\begin{equation*}
u^{Ac}_{p_{r}}(\bl a^{r})=\frac{D(k)\cdot Q_{p_{r},k}}{H_{k}(\bl a^{r})}=\frac{D(k)\cdot S_k(c)}{W_{k}(c)+1}.
\end{equation*}
\end{proof}

\subsection{Proof of Theorem \ref{thm:action-targeted}}
To ease presentation of the proof, throughout this subsection we re-index the topics according to the following order
\[
D(1)\cdot S_1(c) \geq D(2)\cdot  S_2(c)\geq \ldots \geq D(m)\cdot  S_m(c).
\]
The proof of Theorem \ref{thm:action-targeted} relies on several supporting lemmas, which are proven first.
\begin{lemma}
\label{lemma:cyclesactiontargeted}
If $c=(\bl a^1,\ldots,\bl a^l=\bl a^1)$ is an improvement cycle, then for every improvement step $r$ and every topic $k$ it holds that $B_k(\bl a^r)= B_k(\bl a^{r+1})$.
\end{lemma}
\begin{proof}[\textbf{Proof of Lemma \ref{lemma:cyclesactiontargeted}}]
Assume w.l.o.g. that $c$ is a simple improvement cycle. First, we prove by induction on the topic index $k$ that $B_k(\bl a^r)\leq B_k(\bl a^{r+1})$ holds for every $r$, $1\leq r \leq l-1$. Later, we leverage this result to prove the statement of the lemma.

\textbf{Base:} 
By the definition of $S_1(c)$ we know that there exists an improvement step $r'$, $1\leq r'\leq l-1$ such that $B_1(\bl a^{r'})=S_1(c)$. Now Assume that the assertion does not hold for $k=1$; hence, there exists $r''$, $1\leq r'' \leq l-1$, such that $B_1(\bl a^{r''})>B_1(\bl a^{r''+1})$. Therefore, combining the above with the fact that $c$ is an improvement cycle implies that there exists an improvement step $r$ such that $\bl a^r\in \{a^{r'},a^{r'+1},\ldots,a^{r''}\}$ and  $S_1(c)=B_1(\bl a^{r})>B_1(\bl a^{r+1})$. 
As a result, it holds for the improving author $p_r$ in step $r$ that $Q_{p_r,1} =S_1(c)> B_1(\bl a_{-p_r}^r)$ and $H_1(\bl a^r)=1$. In words, the quality of $p_r$'s document exceeds all other qualities under $\bl a^r$ on topic 1; thus,
\begin{equation}
\begin{split}\label{Lemma2:New0}
D(1)\cdot S_1(c)=u^{Ac}_{p_r}(\bl a^{r}).
\end{split}
\end{equation}
In addition, $\bl a^{r+1}$ is an improvement step for author $p_r$, and so $u^{Ac}_{p_r}(\bl a^{r})<u^{Ac}_{p_r}(\bl 
a^{r+1})$. Combined with Equation (\ref{Lemma2:New0}), 
\begin{equation}
\begin{split}\label{Lemma2:New1}
D(1)\cdot S_1(c)<u^{Ac}_{p_r}(\bl a^{r+1}).
\end{split}
\end{equation}
On the other hand, $u^{Ac}_{p_r}(\bl a^{r+1})\leq D(a_{p_r}^{r+1})\cdot S_{a_{p_r}^{r+1}}(c)$ holds; thus, Equation (\ref{Lemma2:New1}) 
implies that $D(1)\cdot S_1(c)< D(a_{p_r}^{r+1})\cdot S_{a_{p_r}^{r+1}}(c)$, which is clearly a contradiction since $D(1)\cdot S_1(c) \geq D(2)\cdot  S_2(c)\geq \ldots \geq D(m)\cdot  S_m(c)$.

\textbf{Step:}  Suppose the assertion holds for every $k$ where $k<K\leq m$, but does not hold for $K$. Similarly to the base case, by the definition of $S_K(c)$, there exists $r'$, $1\leq r'\leq l-1$ such that $B_K(\bl a^{r'})=S_K(c)$. Now since the assertion does not hold for $K$, there exists $r''$, $1\leq r'' \leq l-1$, such that $B_K(\bl a^{r''})>B_K(\bl a^{r''+1})$. Therefore, combining the above with the fact that $c$ is an improvement cycle implies that there exists $r_1$ such that $\bl a^{r_1}\in \{a^{r'},a^{r'+1},\ldots,a^{r''}\}$ and  $S_K(c)=B_K(\bl a^{r_1})>B_K(\bl a^{r_1+1})$.

As a result, it holds for the improving author $p_{r_1}$ in step $r_1$ that $Q_{p_{r_1},K} =S_K(c)> B_K(\bl a_{-p_{r_1}}^{r_1})$ and $H_K(\bl a^{r_1})=1$. In words, the quality of $p_{r_1}$'s document exceeds all other qualities under $\bl a^{r_1}$ on topic K; thus,
\begin{equation}
\begin{split}\label{Lemma2:New3}
D(K)\cdot S_K(c)=u^{Ac}_{p_{r_1}}(\bl a^{r_1}).
\end{split}
\end{equation}

In addition, $u^{Ac}_{p_{r_1}}(\bl a^{r_1})<u^{Ac}_{p_{r_1}}(\bl a^{r_1+1})$ holds since $p_{r_1}$ is the improving author; hence, with Equation (\ref{Lemma2:New3}) we get
\begin{equation}
\begin{split}\label{Lemma2:New4}
D(K)\cdot S_K(c)<u^{Ac}_{p_{r_1}}(\bl a^{r_1+1}).
\end{split}
\end{equation}
Let $k_1$ denote the topic that author $p_{r_1}$ is writing on under $\bl a^{r_1+1}$, i.e $k_1=a^{r_1+1}_{p_{r_1}}$. By definition of $u^{Ac}$ we obtain  
\begin{equation}
\begin{split}\label{Lemma2:New5}
u^{Ac}_{p_{r_1}}(\bl a^{r_1+1})\leq D(k_1)\cdot Q_{p_{r_1},k_1}\leq D(k_1)\cdot S_{k_1}(c) .
\end{split}
\end{equation}
Recall that $D(1)\cdot S_1(c) \geq D(2)\cdot  S_2(c)\geq \ldots \geq D(m)\cdot  S_m(c)$; hence, Equations (\ref{Lemma2:New4}) and (\ref{Lemma2:New5}) suggest that $D(K)\cdot S_K(c)<D(k_1)\cdot S_{k_1}(c)$ holds, and therefore we are guaranteed that $k_1<K$.

Since $k_1<K$, the induction hypothesis hints that  $B_{k_1}(\bl a^{r_1})=B_{k_1}(\bl a^{r_1+1})$; therefore, $Q_{p_{r_1},k_1}\leq B_{k_1}(\bl a^{r_1})$ holds and by Proposition \ref{prop:leqgeqUAC} we get that $Q_{p_{r_1},k_1} = B_{k_1}(\bl a^{r_1})$.Notice that $c$ is a finite improvement path, and that the condition of Proposition \ref{prop:boundonutilityUAC} holds; thus, by invoking it for $r_1,k_1$, we get
\[
u^{Ac}_{p_{r_1}}(\bl a^ {r_1+1})\leq \frac{D(k_1)\cdot S_{k_1}(c)}{W_{k_1}(c)+1}.
\]

Together with Equation (\ref{Lemma2:New4}), we conclude that
\begin{equation}
\label{Lemma2:New6}
D(K)\cdot S_K(c)<\frac{D(k_1)\cdot S_{k_1}(c)}{W_{k_1}(c)+1}.
\end{equation}

Next, we wish to find an improvement step such that the improving author's utility strictly bounds the right-hand-side of Equation 
(\ref{Lemma2:New6}). Since $a_{p_{r_1}}^{r_1+1}=k_1$  and $Q_{p_{r_1},k_1}=B_{k_1}(\bl a^{r_1})$ we get that $H_{k_1}(\bl a^{r_1}) \neq H_{k_1}(\bl a^{r_1+1})$. In addition from the induction hypothesis, we get that for every improvement step $r'$, $B_{k_1}(\bl a^{r'})=S_{k_1}(c)$; hence, we can invoke Proposition \ref{prop:boundonutilityequalUAC}  which guarantees the existence of an index $r_2$ such that $a_{p_{r_2}}^{r_2}=k_1$ 
 and
\begin{equation}
\begin{split}\label{Lemma2:New7}
\frac{D(k_1)\cdot S_{k_1}(c)}{W_{k_1}(c)+1}=u^{Ac}_{p_{r_2}}(\bl a^{r_2}).
\end{split}
\end{equation}
Since $p_{r_2}$ is the improving author $u^{Ac}_{p_{r_2}}(\bl a^{r_2})<u^{Ac}_{p_{r_2}}(\bl a^{r_2+1})$ holds, which together with 
Equation (\ref{Lemma2:New7}) implies
\begin{equation}
\begin{split}\label{Lemma2:New8}
\frac{D(k_1)\cdot S_{k_1}(c)}{W_{k_1}(c)+1}<u^{Ac}_{p_{r_2}}(\bl a^ {r_2+1}).
\end{split}
\end{equation}
Let $a^{r_2+1}_{p_{r_2}}=k_2$. By definition of $u^{Ac}$, we know that
\begin{equation}
\begin{split}\label{Lemma2:New9}
u^{Ac}_{p_{r_2}}(\bl a^{r_2+1})\leq D(k_2)\cdot S_{k_2}(c).
\end{split}
\end{equation}

Observe that $k_2<K$ must hold. To see this, assume otherwise that $k_2\geq K$, and $D(k_2)\cdot S_{k_2}(c)\leq D(K)\cdot S_{K}(c)$ follows. Incorporating this assumption with Equations (\ref{Lemma2:New6}),(\ref{Lemma2:New8}) and (\ref{Lemma2:New9}) we obtain
\[
D(K)\cdot S_{K}(c)<\frac{D(k_1)\cdot S_{k_1}(c)}{W_{k_1}(c)+1}<u^{Ac}_{p_{r_2}}(\bl a^{r_2+1})\leq D(k_2)\cdot S_{k_2}(c) \leq D(K)\cdot S_{K}(c),
\]
which is a contradiction; hence, $k_2<K$. The induction hypothesis hints that  $B_{k_2}(\bl a^{r_2})=B_{k_2}(\bl a^{r_2+1})$, implying $Q_{p_{r_2},k_2}\leq B_{k_2}(\bl a^{r_2})$.

Here again, the condition of Proposition \ref{prop:boundonutilityUAC} holds; thus, by invoking it for $r_2,k_2$ we conclude that
\[
u^{Ac}_{p_{r_2}}(\bl a^ {r_2+1})\leq \frac{D(k_2)\cdot S_{k_2}(c)}{W_{k_2}(c)+1}.
\]
Together with Equation (\ref{Lemma2:New8}), we conclude that
\[
\frac{D(k_1)\cdot S_{k_1}(c)}{W_{k_1}(c)+1}<\frac{D(k_2)\cdot S_{k_2}(c)}{W_{k_2}(c)+1}.
\]
We have therefore bound the right-hand-side of Equation (\ref{Lemma2:New6}) as desired.

This process can be extended to obtain additional $k_3,k_4,\ldots,k_{K}$, such that for all $i\in[K]$, $k_i<K$ and 
\[
\frac{D(k_{1})\cdot S_{k_1}(c)}{W_{k_{1}}(c)+1}<\frac{D(k_{2})\cdot S_{k_2}(c)}{W_{k_{2}}(c)+1}<\frac{D(k_{3})\cdot S_{k_3}(c)}{W_{k_{3}}(c)+1}< \ldots < \frac{D(k_{K})\cdot S_{k_K}(c)}{W_{k_{K}}(c)+1}.
\]
While the inequality above contains $K$ elements, there are only $K-1$ topics with index lower than $K$; hence, at least two of them must be identical, and we obtain a contradiction. We deduce that $B_K(\bl a^{r})\leq B_K(\bl a^{r+1})$ for every step $r$.

This concludes the proof of the induction. Ultimately, to end the proof of this lemma, fix a topic $k$. Due to the induction above, $B_k(\bl a^r)\leq B_k(\bl a^{r+1})$ holds for every $1\leq r \leq l-1$, i.e.,
\[
B_k(\bl a^1)\leq B_k(\bl a^2)\leq\ldots\leq B_k(\bl a^{l-1})\leq B_k(\bl a^l)=B_k(\bl a^1).
\]
The left-hand-side and the right-hand-side of the inequality above are identical; thus, they must all hold in equality. This concludes the proof of this lemma.
\end{proof}

In addition,
\begin{lemma}
\label{lemma:changeinqualitsexpUAC}
If $c=(\bl a^1,\ldots,\bl a^l=\bl a^1)$ is an improvement cycle, then for every improvement step $r$ and topic $k$ such that $a_{p_r}^{r+1}=k$ there exist $(r',k')$ such that $a_{p_{r'}}^{r'+1}=k'$ and
\[
\frac{D(k)\cdot S_{k}(c)}{W_k(c)+1}<\frac{D(k')\cdot S_{k'}(c)}{W_{k'}(c)+1}.
\]
\end{lemma}
\begin{proof}[\textbf{Proof of Lemma \ref{lemma:changeinqualitsexpUAC}}]
Let $r,k$ such that $a^{r+1}_{p_r}=k$. From Lemma \ref{lemma:cyclesactiontargeted} we know that for every improvement step $r''$, $B_k(\bl a^{r''})=S_k(c)$; thus, $Q_{p_r,k}\leq B_k(\bl a^r)$ which by Proposition \ref{prop:leqgeqUAC} leads to
\begin{equation}
\label{newLemma4}
Q_{p_r,k} = B_k(\bl a^r)=S_k(c).
\end{equation}
By definition of improvement step  $a_{p_r}^{r}\neq k$; hence together with Equation (\ref{newLemma4}) we get that $H_k(\bl a^{r})\neq H_k(\bl a^{r+1}).$  Notice that $c$ is a finite improvement path, and that the condition of Proposition \ref{prop:boundonutilityequalUAC} holds; hence, by invoking it for $r,k$ we conclude the existence of an index $r'$ such that $a_{p_{r'}}^{r'}=k$ and
\begin{equation*}
\frac{D(k)\cdot S_{k}(c)}{W_{k}(c)+1}=u^{Ac}_{p_{r'}}(\bl a^ {r'}).
\end{equation*}
In addition, $p_{r'}$ is the improving author, and so
\begin{equation}
\begin{split}\label{Lemma4add}
\frac{D(k)\cdot S_{k}(c)}{W_{k}(c)+1}=u^{Ac}_{p_{r'}}(\bl a^ {r'})<u^{Ac}_{p_{r'}}(\bl a^ {r'+1}).
\end{split}
\end{equation}
Clearly, $ a^{r'+1}_{p_{r'}}=k'\neq k$. Lemma \ref{lemma:cyclesactiontargeted} indicates that $B_{k'}(\bl a^{r'})=B_{k'}(\bl a^{r'+1})$; hence, $Q_{p_{r'},k'}\leq B_{k'}(\bl a^{r'})$. Having showed the condition of Proposition \ref{prop:boundonutilityUAC} holds, we invoke it for $r',k'$ and conclude that
\[
u^{Ac}_{p_{r'}}(\bl a^ {r'+1})\leq  \frac{D(k')\cdot S_{k'}(c)}{W_{k'}(c)+1}.
\]
Combining this fact with Equation (\ref{Lemma4add}), we get
\[
\frac{D(k)\cdot S_{k}(c)}{W_k(c)+1}<\frac{D(k')\cdot S_{k'}(c)}{W_{k'}(c)+1}.
\]
\end{proof}

We are now ready to prove Theorem \ref{thm:action-targeted}.
\begin{proof}[\textbf{Proof of Theorem \ref{thm:action-targeted}}]
Similarly to Theorem \ref{thm:exposure-targeted}, to show that every better-response dynamics converges it suffices to show that every improvement path is finite. Moreover, every improvement path cannot contain more than a finite number of different strategy profiles, as $m,n$ are finite; therefore, if $\gamma$ is infinite it must contain an improvement cycle. We are left to prove that $\gamma$ cannot contain an improvement cycle.

Assume by contradiction that $\gamma$ contains an improvement cycle $c=(\bl a^1,\bl a^2,\ldots,\bl a^l=\bl a^1)$. Let $r_1$ be an arbitrary improvement step and denote by $k_1$ the topic such that $a_{p_{r_1}}^{r_1+1}=k_1$.
From Lemma \ref{lemma:changeinqualitsexpUAC} we know that there exist  $(r_2,k_2)$ such that $a^{r_2+1}_{p_{r_2}}=k_2$ and
\[
\frac{D(k_1)\cdot S_{k_1}(c)}{W_{k_1}(c)+1}<\frac{D(k_2)\cdot S_{k_2}(c)}{W_{k_2}(c)+1}.
\]
Since $a^{r_2+1}_{p_{r_2}}=k_2$, we can now use Lemma \ref{lemma:changeinqualitsexpUAC} again in order to find $(r_3,k_3)$ such that $a^{r_3+1}_{p_{r_3}}=k_3$ and
\[
\frac{D(k_2)\cdot S_{k_2}(c)}{W_{k_2}(c)+1}<\frac{D(k_3)\cdot S_{k_3}(c)}{W_{k_3}(c)+1}.
\]
This process can be extended to achieve additional $k_4,k_5,\ldots, k_{m+1}$ such that
\[
\frac{D(k_1)\cdot S_{k_1}(c)}{W_{k_1}(c)+1}<\frac{D(k_2)\cdot S_{k_2}(c)}{W_{k_2}(c)+1}
<\ldots<\frac{D(k_{m+1})\cdot S_{k_{m+1}}(c)}{W_{k_{m+1}}(c)+1}.
\]
Since there are only $m$ topics and that the inequality above contains $m+1$ elements, there are at least two elements which are identical; thus we obtain a contradiction. We deduce that an improvement cycle can not exist. 

The above suggests that every better-response dynamics must converge.
\end{proof}

\section{Omitted Proofs from Section \ref{sec:nonlearnability}}

\begin{proof}[Proof of Proposition \ref{prop:randislearnable}]
To prove that $R^{RAND}$ is $u^{Ex}$-learnable, one must show that 
every game induced by $R^{RAND}$ and the utility function $u^{Ex}$ has the FIP property. 

Let $G=\langle N,M,D,Q,R^{RAND}, u^{Ex}\rangle$ be an arbitrary game. We now reduce $G$ to a game $G'$ with $R^{PRP}$ as the mediator, where the utility of any author under any strategy profile in $G$ equals to her utility under the same strategy profile in $G'$. If this holds, then every improvement step in $G$ is also an improvement step in $G'$; hence, if there are no improvement cycles in $G'$, then there can be no improvement cycles under $G$ either. Let $G'=\langle N,M,D,Q',R^{PRP}, u^{Ex}\rangle$ for $Q'$ such that
\[
\forall j\in N,k\in M: Q'_{j,k}=1.
\]
Since both $G,G'$ consists of the exposure-targeted utility function, we omit the super-script $Ex$ and use the super-script $G$ to specify the utility of author $j$ under the strategy profile $\bl a$ in $G$, i.e., $u^G_j(\bl a)$, and equivalently for $u^{G'}_j(\bl a)$ for $G'$.

By definition of exposure-targeted utility and $R^{PRP}$, for every valid $j$ and $\bl a$ it holds that
\begin{align*}
u^{G'}_j(\bl a)&= \sum_{k=1}^m \ind_{a_j=k} \cdot D(k)\cdot R^{PRP}_j(Q',k,\bl a)= D(a_j)\cdot R^{PRP}_j(Q',a_j,\bl a)\\
&= D(a_j)\cdot \frac{1}{H_{a_j}(\bl a)}=D(a_j)\cdot R^{RAND}_j(Q,a_j,\bl a)\\
&= \sum_{k=1}^m \ind_{a_j=k} \cdot D(k)\cdot R^{RAND}_j(Q,k,\bl a)\\
& =u^{G}_j(\bl a).
\end{align*}
Since $G'$ possesses $R^{PRP}$ as the mediator, Theorem \ref{thm:exposure-targeted} guarantees that $G'$ has the FIP property. Since we showed $G$ and $G'$ are strategically equivalent, $G$ also has the FIP property, and in particular does not contain improvement cycles.
\end{proof}

\begin{proof}[\textbf{Proof of Theorem \ref{thm:exposure-targeted NL proof}}]
It is sufficient to show that for every $f$ that satisfies the theorem's conditions, we can find a game instance with an improvement cycle. While all it takes to prove the theorem is to construct a \textit{single} counter example (and this is what we do), using the technique below we can actually construct an infinite number of games which do not possess FIP.

Let $R^f$ be a scoring mediator with the corresponding function $f$, which we assume exhibits $f(1)> 2f(0)$. Due to the Intermediate Value Theorem, there exist $x_1,x_2,x_3$ such that $0<x_3<x_2<x_1\leq 1$ and 
\[
\frac{f(x_2)}{f(x_3)}>\frac{2f(x_1)}{f(x_2)}>2.
\]
For brevity, denote $c_1=\frac{f(x_1)}{f(x_2)}$ and $c_2=\frac{f(x_2)}{f(x_3)}$, and observe that $c_2>2 c_1$.

Consider a game with $\abs{N}=4$ authors, $\abs{M}=3$ topics and a quality matrix $Q$ such that 
\[
\begin{pmatrix}
x_1 & 0 & 0 \\
x_1 & x_2 & 0 \\
x_2 & 0 & x_3\\
0 & x_3 & x_2

\end{pmatrix}.
\]
The only missing ingredient is the distribution $D$ over the topics. The selection of such $D$ is crucial: we shall select $D$ to allow improvement cycles. In service of that, we prove the following claim.
\begin{claim}
\label{claim:insidetheoremcycles}
There exists $\epsilon$ such that $0<\epsilon\leq \frac{1}{4}$ and the following properties hold
\begin{enumerate}
\item \label{propertyone} $\frac{c_1(1+c_2)}{c_2(1+2c_1)}<\frac{1-2\epsilon}{2}$,
\item \label{propertytwo} $\frac{1}{1+c_1}<\frac{1-4\epsilon}{2}$, and
\item \label{propertythr} $1<c_2(1-4\epsilon)$.
\end{enumerate}
\end{claim} 
The proof of Claim \ref{claim:insidetheoremcycles} appears after this proof. Now, let $\epsilon$ be an arbitrary constant satisfying the properties of Claim \ref{claim:insidetheoremcycles}, and define $D$ such that 
\[
D(1)=\frac{1}{2-3\epsilon}, \quad D(2)=\frac{1-2\epsilon}{2(2-3\epsilon)} , \quad D(3)=\frac{1-4\epsilon}{2(2-3 \epsilon)}.
\]
It can be verified that $D$ is a valid distribution over the set of topics.

We claim that the game we constructed above possesses an improvement cycle. Consider the strategy profiles
\begin{gather*}
\bl a^1=(1,1,1,2), \quad
\bl a^2=(1,2,1,2), \quad
\bl a^3=(1,2,3,2), \\
\bl a^4=(1,2,3,3), \quad
\bl a^5=(1,1,3,3), \quad
\bl a^6=(1,1,1,3).
\end{gather*}
In the rest of this proof we show that the cycle $c=(\bl a^1,\bl a^2,\bl a^3, \bl a^4,\bl a^5,\bl a^6, \bl a^7=\bl a^1)$ is an improvement cycle. More precisely,  we prove that for every $r$, $1\leq r \leq 6$, $u^{Ex}_{p_r}(\bl a^r)<u^{Ex}_{p_r}(\bl a^{r+1})$.

$\bullet$ $u^{Ex}_{p_1}(\bl a^1)<u^{Ex}_{p_1}(\bl a^{2})$: the deviating author is $p_1 = 2$. Observe that $R_2^f(Q,1,\bl a^1)=\frac{c_1 f(x_2)}{(1+2 c_1) f(x_2)}=\frac{c_1}{1+2 c_1 }$ and $R_2^f(Q,2,\bl a^2)=\frac{c_2  f(x_3)}{(1+c_2) f(x_3)}=\frac{c_2}{1+ c_2 }$. It holds that
\begin{align*}
u^{Ex}_2(\bl a^1)&=R_2^f(Q,1,\bl a^1)\cdot D(1)=\frac{c_1}{1+2\ c_1}\cdot \frac{1}{2-3\epsilon}=\frac{c_2 }{(1+c_2)(2-3 \epsilon)}\cdot \frac{c_1(1+c_2) }{c_2(1+2 c_1) }\\
&\stackrel{\text{Property \ref{propertyone}}}{<}
\frac{c_2 }{(1+c_2)(2-3 \epsilon)}\cdot \frac{1-2\epsilon}{2}
= R_2^f(Q,2,\bl a^2) \cdot D(2)=u^{Ex}_2(\bl a^2);
\end{align*}
thus, $\bl a^2$ is an improvement step.

$\bullet$ $u^{Ex}_{p_2}(\bl a^2)<u^{Ex}_{p_2}(\bl a^{3})$: the deviating author is $p_2 = 3$. Observe that $R_3^f(Q,1,\bl a^2)=\frac{1}{1+c_1}$ and $R_3^f(Q,3,\bl a^3)=1$. It holds that
\begin{align*}
u^{Ex}_3(\bl a^2)&= R_3^f(Q,1,\bl a^2) \cdot D(1)=\frac{1}{1+c_1}\cdot \frac{1}{2-3 \epsilon}\stackrel{\text{Property \ref{propertytwo}}}{<}\frac{1-4\epsilon}{2}\cdot \frac{1}{2-3\epsilon}\\&=1\cdot D(3)=R_3^f(Q,3,\bl a^3)\cdot D(3)=u^{Ex}_3(\bl a^3);
\end{align*}
thus, $\bl a^3$ is an improvement step.

$\bullet$ $u^{Ex}_{p_3}(\bl a^3)<u^{Ex}_{p_3}(\bl a^{4})$: the deviating author is $p_3 = 4$. Observe that $R_4^f(Q,2,\bl a^3)=\frac{1}{1+c_2}$ and $R_4^f(Q,3,\bl a^4)=\frac{c_2}{1+c_2}$. It holds that
\begin{align*}
u^{Ex}_4(\bl a^3)&=R_4^f(Q,2,\bl a^3) \cdot D(2)= \frac{1}{1+c_2} \cdot \frac{1-2\epsilon}{2(2-3 \epsilon)}\stackrel{\text{Property \ref{propertythr}}}{<}\frac{c_2}{1+c_2}\cdot \frac{1-4\epsilon}{2(2-3 \epsilon)}\\& =R_4^f(Q,3,\bl a^4)\cdot D(3)=u^{Ex}_4(\bl a^4);
\end{align*}
thus, $\bl a^4$ is an improvement step.

$\bullet$ $u^{Ex}_{p_4}(\bl a^4)<u^{Ex}_{p_4}(\bl a^{5})$: the deviating author is $p_4 = 2$. Observe that $R_2^f(Q,2,\bl a^4)=1$ and $R_2^f(Q,1,\bl a^5)=\frac{1}{2}$. It holds that
\begin{align*}
u^{Ex}_2(\bl a^4)&=R_2^f(Q,2,\bl a^4) \cdot D(2)=1\cdot D(2)=\frac{1-2\epsilon}{2(2-3 \epsilon)}<\frac{1}{2}\cdot \frac{1}{2-3 \epsilon}\\&=\frac{1}{2} D(1) =R_2^f(Q,1,\bl a^5)\cdot D(1)=u^{Ex}_2(\bl a^5);
\end{align*}
thus, $\bl a^5$ is an improvement step.

$\bullet$ $u^{Ex}_{p_5}(\bl a^5)<u^{Ex}_{p_5}(\bl a^{6})$: the deviating author is $p_5 = 3$. Observe that $R_3^f(Q,3,\bl a^5)=\frac{1}{1+c_2}$ and $R_3^f(Q,1,\bl a^6)=\frac{1}{1+2c_1}$. It holds that
\begin{align*}
u^{Ex}_3(\bl a^5)&=R_3^f(Q,3,\bl a^5) \cdot D(3)=\frac{1}{1+c_2}\cdot \frac{1-4\epsilon}{2(2-3 \epsilon)}<\frac{1}{1+c_2}\cdot \frac{1}{2-3 \epsilon}\\&\stackrel{2 c_1< c_2 }{<}\frac{1}{1+2 c_1}\cdot\frac{1}{2-3 \epsilon} =R_3^f(Q,1,\bl a^6)\cdot D(1)=u^{Ex}_3(\bl a^6);
\end{align*}
thus, $\bl a^6$ is an improvement step.

$\bullet$ $u^{Ex}_{p_6}(\bl a^6)<u^{Ex}_{p_6}(\bl a^{1})$: the deviating author is $p_6 = 4$. Observe that $R_4^f(Q,3,\bl a^6)=1$ and $R_4^f(Q,2,\bl a^1)=1$. It holds that
\begin{align*}
u^{Ex}_4(\bl a^6)&=R_4^f(Q,3,\bl a^6) \cdot D(3)=1\cdot D(3)=\frac{1-4\epsilon}{2(2-3 \epsilon)}<\frac{1-2\epsilon}{2(2-3 \epsilon)}\\&=1\cdot D(2)=R_4^f(Q,2,\bl a^1)\cdot D(2)=u^{Ex}_4(\bl a^1);
\end{align*}
thus, $\bl a^1$ is an improvement step.

The above analysis implies that  $c=(\bl a^1,\bl a^2,\bl a^3, \bl a^4,\bl a^5,\bl a^6, \bl a^7=\bl a^1)$ in an improvement cycle. As a result, $R^f$ is not $u^{Ex}$-learnable.

\end{proof}
\begin{proof}[\textbf{Proof of Claim \ref{claim:insidetheoremcycles}}]
Since $c_2>2c_1$, it follows that
\[
\frac{c_1(1+c_2) }{c_2(1+2 c_1) }=\frac{c_1+c_1 c_2}{c_2+2 c_1 c_2}<\frac{\frac{c_2}{2}+c_1 c_2}{c_2+2 c_1 c_2}=\frac{1}{2}\cdot \frac{c_2+2 c_1 c_2}{c_2+2 c_1 c_2}=\frac{1}{2}.
\]
Since the left-hand-side is strictly less than $\frac{1}{2}$, we denote by $\epsilon_1$ a positive real number such that
\[
\frac{c_1(1+c_2)}{c_2(1+2 c_1) }<\frac{1-2\epsilon_1}{2}.
\]
In addition, notice that $c_1>1$; thus, $\frac{1}{1+c_1}<\frac{1}{2}$. We denote by $\epsilon_2$ a positive real number such that
\[
\frac{1}{1+c_1}<\frac{1-4\cdot \epsilon_2}{2}.
\]
Similarly, since $c_2$ is constant and $c_2>2$, there exists $\epsilon_3>0$ such that  $c_2\cdot \epsilon_3< \frac{1}{5}$, which implies that
\[
c_2\cdot (1-4\epsilon_3)>c_2-1>1. 
\]
The proof is completed by setting $\epsilon=\min\{\epsilon_1,\epsilon_2,\epsilon_3\}$.

\end{proof}

\begin{proof}[\textbf{Proof of Theorem \ref{thm:action-targeted NL proof11}}]

It is sufficient to show that for every $f$ that satisfies the theorem's conditions, we can find a game instance with an improvement cycle. While all it takes to prove the theorem is to construct a \textit{single} counter example (and this is what we do), using the technique below we can actually construct an infinite number of games which do not possess FIP.

Let $R^f$ be a scoring mediator with the corresponding function $f$, which we assume exhibits ${f(1)>2(2\alpha -1) f\left(\frac{1}{\alpha}\right)}$
for some $\alpha >1$. Due to the Intermediate Value Theorem, there exist $x_1,x_2,x_3$ such that $\frac{1}{\alpha}<x_3<x_2<x_1\leq 1$ and 
\[
2(2\alpha-1)<\frac{2f(x_1)}{f(x_2)}<\frac{f(x_1)}{f(x_3)}<2(2\alpha-\frac{1}{2})
\]
For brevity, denote $c_1=\frac{f(x_1)}{f(x_2)}$ and $c_2=\frac{f(x_1)}{f(x_3)}$, and observe that $c_2>2 c_1$.

Consider a game with $\abs{N}=4$ authors, $\abs{M}=3$ topics and a quality matrix $Q$ such that 
\[
\begin{pmatrix}
x_1 & 0 & 0 \\
x_1 & x_1 & 0 \\
x_2 & 0 & x_2\\
0 & x_3 & x_2

\end{pmatrix}.
\]
The only missing ingredient is the distribution $D$ over the topics. The selection of such $D$ is crucial: we shall select $D$ to allow improvement cycles. In service of that, we prove the following claim.
\begin{claim}
\label{claim:insidetheoremcyclesofThm4}
There exists $\epsilon$ such that $0<\epsilon\leq \frac{1}{6}$ and the following properties hold
\begin{enumerate}
\item \label{propertyoneThm4} $\frac{c_1(1+c_2)}{c_2(1+2c_1)}<\frac{1-\epsilon}{2}$,
\item \label{propertytwoThm4} $\frac{1}{1+c_1}<\frac{1-2\epsilon}{2\alpha}$, and
\item \label{propertythrThm4} $\frac{\alpha(1-\epsilon)}{1+c_2}<\frac{1-2\epsilon}{2}$.
\end{enumerate}
\end{claim} 
The proof of Claim \ref{claim:insidetheoremcyclesofThm4} appears after this proof. Now, let $\epsilon$ be an arbitrary constant satisfying the properties of Claim \ref{claim:insidetheoremcyclesofThm4}, and define $D$ such that 
\[
D(1)=\frac{2\alpha}{3\alpha+1-\epsilon(\alpha+2)}, \quad D(2)=\frac{\alpha(1-\epsilon)}{3\alpha+1-\epsilon(\alpha+2)} , \quad D(3)=\frac{1-2\epsilon}{3\alpha+1-\epsilon(\alpha+2)}.
\]
It can be verified that $D$ is a valid distribution over the set of topics.

We claim that the game we constructed above possesses an improvement cycle. Consider the strategy profiles
\begin{gather*}
\bl a^1=(1,1,1,2), \quad
\bl a^2=(1,2,1,2), \quad
\bl a^3=(1,2,3,2), \\
\bl a^4=(1,2,3,3), \quad
\bl a^5=(1,1,3,3), \quad
\bl a^6=(1,1,1,3).
\end{gather*}
In the rest of this proof we show that the cycle $c=(\bl a^1,\bl a^2,\bl a^3, \bl a^4,\bl a^5,\bl a^6, \bl a^7=\bl a^1)$ is an improvement cycle. More precisely,  we prove that for every $r$, $1\leq r \leq 6$, $u^{Ac}_{p_r}(\bl a^r)<u^{Ac}_{p_r}(\bl a^{r+1})$.

$\bullet$ $u^{Ac}_{p_1}(\bl a^1)<u^{Ac}_{p_1}(\bl a^{2})$: the deviating author is $p_1 = 2$. Observe that $R_2^f(Q,1,\bl a^1)=\frac{c_1 f(x_2)}{(1+2 c_1) f(x_2)}=\frac{c_1}{1+2 c_1 }$ and $R_2^f(Q,2,\bl a^2)=\frac{c_2  f(x_3)}{(1+c_2) f(x_3)}=\frac{c_2}{1+ c_2 }$. It holds that
\begin{align*}
u^{Ac}_2(\bl a^1)&=R_2^f(Q,1,\bl a^1)\cdot D(1)\cdot Q_{2,1}=\frac{c_1}{1+2c_1}\cdot \frac{2\alpha}{3\alpha+1-\epsilon(\alpha+2)}\cdot x_1\\&=\frac{2\alpha c_2}{(1+c_2)\left(3\alpha+1-\epsilon(\alpha+2)\right)}\cdot\frac{c_1(1+c_2)}{c_2(1+2\ c_1)}\cdot  x_1
\\& \stackrel{\text{Property \ref{propertyoneThm4}}}{<}\frac{2\alpha c_2}{(1+c_2)\left(3\alpha+1-\epsilon(\alpha+2)\right)}\cdot\frac{1-\epsilon}{2}\cdot x_1=\frac{c_2}{1+c_2}\cdot \frac{\alpha(1-\epsilon)}{3\alpha+1-\epsilon(\alpha+2)}\cdot x_1\\&= R_2^f(Q,2,\bl a^2) \cdot D(2)\cdot Q_{2,2}=u^{Ac}_2(\bl a^2);
\end{align*}
thus, $\bl a^2$ is an improvement step.

$\bullet$ $u^{Ac}_{p_2}(\bl a^2)<u^{Ac}_{p_2}(\bl a^{3})$: the deviating author is $p_2 = 3$. Observe that $R_3^f(Q,1,\bl a^2)=\frac{1}{1+c_1}$ and $R_3^f(Q,3,\bl a^3)=1$. It holds that
\begin{align*}
u^{Ac}_3(\bl a^2)&= R_3^f(Q,1,\bl a^2) \cdot D(1)\cdot Q_{3,1}=\frac{1}{1+c_1}\cdot\frac{2\alpha}{3\alpha+1-\epsilon(\alpha+2)}\cdot x_2\\&\stackrel{\text{Property \ref{propertytwoThm4}}}{<}\frac{1-2\epsilon}{2\alpha}\cdot\frac{2\alpha}{3\alpha+1-\epsilon(\alpha+2)}\cdot x_2=\frac{1-2\epsilon}{3\alpha+1-\epsilon(\alpha+2)}\cdot x_2\\&=1\cdot D(3)\cdot x_2=R_3^f(Q,3,\bl a^3)\cdot D(3)\cdot Q_{3,3}=u^{Ac}_3(\bl a^3);
\end{align*}
thus, $\bl a^3$ is an improvement step.

$\bullet$ $u^{Ac}_{p_3}(\bl a^3)<u^{Ac}_{p_3}(\bl a^{4})$: the deviating author is $p_3 = 4$. Observe that $R_4^f(Q,2,\bl a^3)=\frac{1}{1+c_2}$ and $R_4^f(Q,3,\bl a^4)=\frac{1}{2}$. It holds that
\begin{align*}
u^{Ac}_4(\bl a^3)&=R_4^f(Q,2,\bl a^3) \cdot D(2)\cdot Q_{4,2}= \frac{1}{1+c_2} \cdot \frac{\alpha(1-\epsilon)}{3\alpha+1-\epsilon(\alpha+2)}\cdot x_3 \\&\stackrel{\text{Property \ref{propertythrThm4}}}{<}\frac{1-2\epsilon}{2}\cdot \frac{1}{3\alpha+1-\epsilon(\alpha+2)}\cdot x_3\stackrel{x_3<x_2}{<}\frac{1}{2}\cdot \frac{1-2\epsilon}{3\alpha+1-\epsilon(\alpha+2)}\cdot x_2\\& =R_4^f(Q,3,\bl a^4)\cdot D(3)\cdot Q_{4,3}=u^{Ac}_4(\bl a^4);
\end{align*}
thus, $\bl a^4$ is an improvement step.

$\bullet$ $u^{Ac}_{p_4}(\bl a^4)<u^{Ex}_{p_4}(\bl a^{5})$: the deviating author is $p_4 = 2$. Observe that $R_2^f(Q,2,\bl a^4)=1$ and $R_2^f(Q,1,\bl a^5)=\frac{1}{2}$. It holds that
\begin{align*}
u^{Ac}_2(\bl a^4)&=R_2^f(Q,2,\bl a^4) \cdot D(2)\cdot Q_{2,2}=1\cdot D(2)\cdot x_2=\frac{\alpha(1-\epsilon)}{3\alpha+1-\epsilon(\alpha+2)}\cdot x_2\\&<\frac{1}{2}\cdot \frac{2\alpha}{3\alpha+1-\epsilon(\alpha+2)}\cdot x_2=\frac{1}{2} D(1)\cdot x_2 =R_2^f(Q,1,\bl a^5)\cdot D(1)\cdot Q_{2,1}=u^{Ac}_2(\bl a^5);
\end{align*}
thus, $\bl a^5$ is an improvement step.

$\bullet$ $u^{Ac}_{p_5}(\bl a^5)<u^{Ac}_{p_5}(\bl a^{6})$: the deviating author is $p_5 = 3$. Observe that $R_3^f(Q,3,\bl a^5)=\frac{1}{2}$ and $R_3^f(Q,1,\bl a^6)=\frac{1}{1+2c_1}$. It holds that
\begin{align*}
u^{Ac}_3(\bl a^5)&=R_3^f(Q,3,\bl a^5) \cdot D(3)\cdot Q_{3,3}=\frac{1}{2}\cdot \frac{1-2\epsilon}{3\alpha+1-\epsilon(\alpha+2)}\cdot x_2\\&=\frac{1}{2}\cdot \frac{1-2\epsilon}{3\alpha+1-\epsilon(\alpha+2)}\cdot\frac{4\alpha}{4\alpha}\cdot x_2\stackrel{2c_1<4\alpha-1}{<}\frac{1}{2}\cdot \frac{1-2\epsilon}{3\alpha+1-\epsilon(\alpha+2)}\cdot\frac{4\alpha}{1+2c_1}\cdot x_2\\&<\frac{1}{1+2c_1}\cdot\frac{2\alpha}{3\alpha+1-\epsilon(\alpha+2)}\cdot x_2 =R_3^f(Q,1,\bl a^6)\cdot D(1)\cdot Q_{3,1}=u^{Ac}_3(\bl a^6);
\end{align*}
thus, $\bl a^6$ is an improvement step.

$\bullet$ $u^{Ac}_{p_6}(\bl a^6)<u^{Ac}_{p_6}(\bl a^{1})$: the deviating author is $p_6 = 4$. Observe that $R_4^f(Q,3,\bl a^6)=1$ and $R_4^f(Q,2,\bl a^1)=1$.Since $\frac{1}{\alpha}<x_3<x_2<1$ we get that
\[
\frac{x_2}{x_3} < \alpha,
\]
implying that $x_2<\alpha x_3$; thus, it holds that
\begin{align*}
u^{Ac}_4(\bl a^6)&=R_4^f(Q,3,\bl a^6) \cdot D(3)\cdot Q_{4,3}=\frac{1-2\epsilon}{3\alpha+1-\epsilon(\alpha+2)}\cdot x_2\\&<\frac{\alpha(1-2\epsilon)}{3\alpha+1-\epsilon(\alpha+2)}\cdot  x_3<\frac{\alpha(1-\epsilon)}{3\alpha+1-\epsilon(\alpha+2)}\cdot  x_3\\&=R_4^f(Q,2,\bl a^1)\cdot D(2)\cdot Q_{4,2}=u^{Ac}_4(\bl a^1);
\end{align*}
thus, $\bl a^1$ is an improvement step.

The above analysis implies that  $c=(\bl a^1,\bl a^2,\bl a^3, \bl a^4,\bl a^5,\bl a^6, \bl a^7=\bl a^1)$ in an improvement cycle. As a result, $R^f$ is not $u^{Ac}$-learnable.

\end{proof}
\begin{proof}[\textbf{Proof of Claim \ref{claim:insidetheoremcyclesofThm4}}]
%We denote $c_1=\frac{f(x_1)}{f(x_2)}, \quad c_2=\frac{f(x_2)}{f(x_3)}$; thus, $c_2>2\cdot c_1$. By the way we chose $c_1,c_2$, we know that $c_1>1, c_2>2$.
Since $c_2>2c_1$, it follows that
\[
\frac{c_1(1+c_2) }{c_2(1+2 c_1) }=\frac{c_1+c_1 c_2}{c_2+2 c_1 c_2}<\frac{\frac{c_2}{2}+c_1 c_2}{c_2+2 c_1 c_2}=\frac{1}{2}\cdot \frac{c_2+2 c_1 c_2}{c_2+2 c_1 c_2}=\frac{1}{2}.
\]
Since the left-hand-side is strictly less than $\frac{1}{2}$, we denote by $\epsilon_1$ a positive real number such that
\[
\frac{c_1(1+c_2)}{c_2(1+2 c_1) }<\frac{1-\epsilon_1}{2}.
\]
In addition, notice that $c_1>2\alpha -1$; thus, $\frac{1}{1+c_1}<\frac{1}{2\alpha}$. We denote by $\epsilon_2$ a positive real number such that
\[
\frac{1}{1+c_1}<\frac{1-2 \epsilon_2}{2\alpha}.
\]
Since $\alpha>1$ and $c_2>4\alpha-2$, for every $0<\epsilon_3<\frac{1}{6}$ it holds that
\[
\frac{\alpha(1-\epsilon_3)}{1+c_2}<\frac{\alpha(1-\epsilon_3)}{1+4\alpha -2}=\frac{\alpha(1-\epsilon_3)}{4\alpha -1}<\frac{\alpha}{3\alpha}=\frac{1}{3}=\frac{1}{2}-\frac{1}{6}<\frac{1-2\epsilon_3}{2}.
\]
The proof is completed by setting $\epsilon=\min\{\epsilon_1,\epsilon_2,\epsilon_3\}$.
\end{proof}

\section{Non-Learnability under Action-Targeted Utility}
In this section we prove non-learnability of another family of scoring mediators under $u^{Ac}$.
We consider scoring mediators where the corresponding function $f$ is bounded by (non-affine) linear functions. Examples for such functions are $f(x)=x$ and $f(x)=e^{x}-1$.
\begin{theorem}
\label{thm:action-targeted NL proof1 linear}
Let $R^f$ be a scoring mediator. If $f$ is continuous function and there exist $\alpha,\beta >0$ such that $\alpha x \leq f(x) \leq \beta x, \forall x\in [0,1]$, then $R^f$ is not $u^{Ac}$-learnable. 
\end{theorem}
\begin{proof}[\textbf{Proof of Theorem \ref{thm:action-targeted NL proof1 linear}}]
It is sufficient to show that for every $f$ that satisfies the theorem's conditions, we can find a game instance with an improvement cycle. While all it takes to prove the theorem is to construct a \textit{single} counter example (and this is what we do), using the technique below we can actually construct an infinite number of games which do not possess FIP.

Let $R^f$ be a scoring mediator with the corresponding continuous function $f$, and let $\alpha$ and $\beta$ such that $\alpha,\beta>0$ and
\[
\forall x\in [0,1] \quad \alpha x \leq f(x) \leq \beta x.
\]

For brevity, denote $z=\frac{\beta}{\alpha}$. Since $\alpha \cdot 0 \leq f(0) \leq \beta \cdot 0$ we know that $f(0)=0$. Notice that for every, $x\in (0,1]$, it must hold that $f(x)>0$ as $\alpha x\leq f(x)$ where $\alpha>0$. Let $x_1$ be an arbitrary quality such that $x_1 \in (0,1]$.
Due to the Intermediate Value Theorem, there exist $x_2,x_3$ such that $0<x_3<x_2<x_1\leq 1$ and 
\[
5 z  f(x_2)=f(x_1),\quad 11 z f(x_3)=f(x_1).
\]

Consider a game with $\abs{N}=4$ authors, $\abs{M}=3$ topics and a quality matrix $Q$ such that 
\[
\begin{pmatrix}
x_1 & 0 & 0 \\
x_1 & x_1 & 0 \\
x_2 & 0 & x_2\\
0 & x_3 & x_2

\end{pmatrix}.
\]
The only missing ingredient is the distribution $D$ over the topics. The selection of such $D$ is crucial: we shall select $D$ to allow improvement cycles. In service of that, we prove the following claim.
\begin{claim}
\label{claim:insidetheoremcyclesUACEpsilon}
There exists $\epsilon$ such that $\epsilon>0$ and the following properties hold
\begin{enumerate}
\item \label{propertyoneUAC} $\frac{10z}{10z+1}\cdot \left (\frac{5z+0.6}{15z+3.8}+\frac{\epsilon}{2} \right)<\frac{11z}{11z+1}\cdot \left (\frac{5z+0.6}{15z+3.8}-\epsilon \right)$
%\item \label{propertyoneUAC} $\left(\frac{11z}{11z+1}-\frac{10z}{10z+1}\right)\cdot \left(\frac{5z+0.6}{15z+3.8}\right)>\left(\frac{11z}{11z+1}+\frac{10z}{2(10z+1)}\right)\cdot \epsilon$,
\item \label{propertytwoUAC} $\frac{10z+1.2+\epsilon(15z+3.8)}{5z+1}<2$, and
\item \label{propertythrUAC} $\frac{4.4\cdot z}{15z+3.8}<\frac{5z+0.6}{15z+3.8}-\epsilon$.
%\item \label{propertythrUAC} $\frac{0.6\cdot z+0.6}{15z+3.8}>\epsilon$.
\end{enumerate}
\end{claim} 
The proof of Claim \ref{claim:insidetheoremcyclesUACEpsilon} appears after this proof. Now, let $\epsilon$ be an arbitrary constant satisfying the properties of Claim \ref{claim:insidetheoremcyclesUACEpsilon}, and define $D$ such that
\[
D(1)=\frac{10z+1.2}{15z+3.8}+\epsilon,\quad D(2)=\frac{5z+0.6}{15z+3.8}-\epsilon, \quad D(3)=\frac{2}{15z+3.8}
\]

It can be verified that $D$ is a valid distribution over the set of topics.

We claim that the game we constructed above possesses an improvement cycle. Consider the strategy profiles
\begin{gather*}
\bl a^1=(1,1,1,2), \quad
\bl a^2=(1,2,1,2), \quad
\bl a^3=(1,2,3,2), \\
\bl a^4=(1,2,3,3), \quad
\bl a^5=(1,1,3,3), \quad
\bl a^6=(1,1,1,3).
\end{gather*}
In the rest of this proof we show that the cycle $c=(\bl a^1,\bl a^2,\bl a^3, \bl a^4,\bl a^5,\bl a^6, \bl a^7=\bl a^1)$ is an improvement cycle. More precisely,  we prove that for every $r$, $1\leq r \leq 6$, $u^{Ac}_{p_r}(\bl a^r)<u^{Ac}_{p_r}(\bl a^{r+1})$.

$\bullet$ $u^{Ac}_{p_1}(\bl a^1)<u^{Ac}_{p_1}(\bl a^{2})$: the deviating author is $p_1 = 2$. Observe that $R_2^f(Q,1,\bl a^1)=\frac{5z f(x_2)}{(10z+1) f(x_2)}=\frac{5z}{10z+1}$ and $R_2^f(Q,2,\bl a^2)=\frac{11z  f(x_3)}{(11z+1) f(x_3)}=\frac{11z}{11z+1 }$. It holds that
\begin{align*}
u_2^{Ac}(\bl a^1)&=R^f_2(Q,1,\bl a^1)\cdot D(1) \cdot Q_{2,1} = \frac{5z}{10z+1}\cdot \left (\frac{10z+1.2}{15z+3.8}+\epsilon \right)\cdot x_1\\&=\frac{1}{2}\cdot \frac{10z}{10z+1}\cdot \left (\frac{10z+1.2}{15z+3.8}+\epsilon \right)\cdot x_1=\frac{10z}{10z+1}\cdot \left (\frac{5z+0.6}{15z+3.8}+\frac{\epsilon}{2} \right)\cdot x_1\\&\stackrel{\text{Property \ref{propertyoneUAC}}}{<}\frac{11z}{11z+1}\cdot \left (\frac{5z+0.6}{15z+3.8}-\epsilon \right)\cdot x_1=R^f_2(Q,2,\bl a^2)\cdot D(2) \cdot Q_{2,2}=u_2^{Ac}(\bl a^2);
\end{align*}
thus, $\bl a^2$ is an improvement step.

$\bullet$ $u^{Ac}_{p_2}(\bl a^2)<u^{Ac}_{p_2}(\bl a^{3})$: the deviating author is $p_2 = 3$. Observe that $R_3^f(Q,1,\bl a^2)=\frac{1}{5z+1}$ and $R_3^f(Q,3,\bl a^3)=1$. It holds that
\begin{align*}
u_3^{Ac}(\bl a^2)&=R^f_3(Q,1,\bl a^2)\cdot D(1) \cdot Q_{3,1} = \frac{1}{5z+1}\cdot \left(\frac{10z+1.2}{15z+3.8}+\epsilon \right)\cdot x_2\\&=\frac{1}{15z+3.8}\cdot\left( \frac{10z+1.2+\epsilon(15z+3.8)}{5z+1}\right)\cdot x_2\stackrel{\text{Property \ref{propertytwoUAC}}}{<}\frac{2}{15z+3.8}\cdot x_2\\&=1\cdot D(3)\cdot x_2=R^f_3(Q,3,\bl a^2)\cdot D(2) \cdot Q_{3,3}=u_3^{Ac}(\bl a^3);
\end{align*}
thus, $\bl a^3$ is an improvement step.

$\bullet$ $u^{Ac}_{p_3}(\bl a^3)<u^{Ac}_{p_3}(\bl a^{4})$: the deviating author is $p_3 = 4$. Observe that $R_4^f(Q,2,\bl a^3)=\frac{1}{11z+1}$ and $R_4^f(Q,3,\bl a^4)=\frac{1}{2}$. It holds that
\begin{align*}
u_4^{Ac}(\bl a^3)&=R^f_4(Q,2,\bl a^3)\cdot D(2) \cdot Q_{4,2} = \frac{1}{11z+1}\cdot \left (\frac{5z+0.6}{15z+3.8}-\epsilon \right)\cdot x_3\\&\stackrel{x_3<x_2}{<}\frac{1}{11z+1}\cdot \left (\frac{5z+0.6}{15z+3.8}-\epsilon \right)\cdot x_2<\frac{1}{11z+1}\cdot \frac{5z+0.6}{15z+3.8} \cdot x_2\\&<\frac{1}{15z+3.8} \cdot x_2=\frac{1}{2}\cdot \frac{2}{15z+3.8} \cdot x_2=R^f_4(Q,3,\bl a^4)\cdot D(3) \cdot Q_{4,3}=u_4^{Ac}(\bl a^3);
\end{align*}
thus, $\bl a^4$ is an improvement step.

$\bullet$ $u^{Ac}_{p_4}(\bl a^4)<u^{Ex}_{p_4}(\bl a^{5})$: the deviating author is $p_4 = 2$. Observe that $R_2^f(Q,2,\bl a^4)=1$ and $R_2^f(Q,1,\bl a^5)=\frac{1}{2}$. It holds that
\begin{align*}
u_2^{Ac}(\bl a^4)&=R_2^f(Q,2,\bl a^4)\cdot D(2)\cdot Q_{2,2}=1\cdot D(2)\cdot x_1=\left(\frac{5z+0.6}{15z+3.8}-\epsilon \right)\cdot x_1\\&=\frac{1}{2}\cdot\left(\frac{10z+1.2}{15z+3.8}-2\cdot \epsilon \right)\cdot x_1<\frac{1}{2}\cdot\left(\frac{10z+1.2}{15z+3.8}+ \epsilon \right)\cdot x_1
=\frac{1}{2}\cdot D(1)\cdot x_1 \\&=R_2^f(Q,1,\bl a^5)\cdot D(1)\cdot Q_{1,2}=u_2^{Ac}(\bl a^5);
\end{align*}
thus, $\bl a^5$ is an improvement step.

$\bullet$ $u^{Ac}_{p_5}(\bl a^5)<u^{Ac}_{p_5}(\bl a^{6})$: the deviating author is $p_5 = 3$. Observe that $R_3^f(Q,3,\bl a^5)=\frac{1}{2}$ and $R_3^f(Q,1,\bl a^6)=\frac{1}{10z+1}$. It holds that
\begin{align*}
u_3^{Ac}(\bl a^5)&=R_3^f(Q,3,\bl a^5)\cdot D(3)\cdot Q_{3,3}=\frac{1}{2}\cdot  \frac{2}{15z+3.8} \cdot x_2\\&=\frac{1}{15z+3.8} \cdot x_2<\frac{10z+1.2}{10z+1}\cdot\frac{1}{15z+3.8} \cdot x_2=\frac{1}{10z+1}\cdot\frac{10z+1.2}{15z+3.8} \cdot x_2\\&=R_3^f(Q,1,\bl a^6)\cdot D(1)\cdot Q_{3,1}=u_3^{Ac}(\bl a^6);
\end{align*}
thus, $\bl a^6$ is an improvement step.

$\bullet$ $u^{Ac}_{p_6}(\bl a^6)<u^{Ac}_{p_6}(\bl a^{1})$: the deviating author is $p_6 = 4$. Observe that $R_4^f(Q,3,\bl a^6)=1$ and $R_4^f(Q,2,\bl a^1)=1$.Since for every $x$, $ \alpha x\leq f(x)\leq \beta x$ we get that
\[
\frac{x_2}{x_3} \leq \frac{\frac{f(x_2)}{\alpha x_2}}{\frac{f(x_3)}{\beta x_3}}=\frac{\beta}{\alpha}\cdot\frac{f(x_2)}{f(x_3)}=\frac{11z}{5},
\]
implying that $f(x_2)<\frac{11z}{5} \cdot x_3$; thus, it holds that
\begin{align*}
u_4^{Ac}(\bl a^6)&=R_4^f(Q,3,\bl a^6)\cdot D(3)\cdot Q_{4,3}=1 \cdot \frac{2}{15z+3.8}\cdot x_2\leq \frac{2}{15z+3.8} \cdot \frac{11z}{5}\cdot x_3\\&=\frac{2}{15z+3.8}\cdot 2.2 z\cdot  x_3=\frac{4.4 z}{15z+3.8}\cdot x_3\stackrel{\text{Property \ref{propertythrUAC}}}{<}\left(\frac{5z+0.6}{15z+3.8}-\epsilon\right)\cdot x_3=\\& =R_4^f(Q,2,\bl a^1)\cdot D(2)\cdot Q_{4,2}=u_4^{Ac}(\bl a^1);
\end{align*}
thus, $\bl a^1$ is an improvement step.

The above analysis implies that  $c=(\bl a^1,\bl a^2,\bl a^3, \bl a^4,\bl a^5,\bl a^6, \bl a^7=\bl a^1)$ in an improvement cycle. As a result, $R^f$ is not $u^{Ac}$-learnable.
\end{proof}
\begin{proof}[\textbf{Proof of Claim \ref{claim:insidetheoremcyclesUACEpsilon}}]
%We denote $c_1=\frac{f(x_1)}{f(x_2)}, \quad c_2=\frac{f(x_2)}{f(x_3)}$; thus, $c_2>2\cdot c_1$. By the way we chose $c_1,c_2$, we know that $c_1>1, c_2>2$.
As $\frac{10z}{10z+1}<\frac{11z}{11z+1}$, we can find $\epsilon_1>0$ such that, 
\[
\left(\frac{11z}{11z+1}+\frac{10z}{2(10z+1)}\right)\cdot \epsilon_1.<\left(\frac{11z}{11z+1}-\frac{10z}{10z+1}\right)\cdot \left(\frac{5z+0.6}{15z+3.8}\right);
\]
hence, we get that
\[
\frac{10z}{10z+1}\cdot \left (\frac{5z+0.6}{15z+3.8}+\frac{\epsilon_1}{2} \right)<\frac{11z}{11z+1}\cdot \left (\frac{5z+0.6}{15z+3.8}-\epsilon_1 \right).
\]
Since $10z+1.2<2 (5z+1)$ we can find $\epsilon_2>0$ such that
\[
\frac{10z+1.2}{5z+1}+4\epsilon_2<2.
\]
Therefore,
\[
\frac{10z+1.2+\epsilon_2(15z+3.8)}{5z+1}=\frac{10z+1.2}{5z+1}+\epsilon_2 \cdot \frac{15z+3.8}{5z+1}<\frac{10z+1.2}{5z+1}+4\epsilon_2<2
\]
In addition, we can find $\epsilon_3>0$ such that
\[
\epsilon_3<\frac{0.6 z+0.6}{15z+3.8},
\]
which implies that 
\[
\frac{4.4 z}{15z+3.8}<\frac{5z+0.6}{15z+3.8}-\epsilon_3.
\]
The proof is completed by setting $\epsilon=\min\{\epsilon_1,\epsilon_2,\epsilon_3\}$.
\end{proof}

}\fi} %closing appendices

\end{document}